\documentclass[acmsmall,authorversion,nonacm]{acmart}
\PassOptionsToPackage{hyphens}{url}
\usepackage{paralist}
\usepackage{booktabs}
\usepackage{xspace}
\usepackage{hyperref}
\usepackage{cleveref}
\usepackage{adjustbox}
\usepackage{pifont}
\usepackage{array}
\usepackage{xcolor}
\usepackage{multirow}
\usepackage{enumitem}
\usepackage{xurl}
\usepackage{subfig}

\newcommand{\xclaim}{\textsc{Xclaim}\xspace}
\newcommand{\name}{\textsc{CroCoDai}\xspace}

\title{\name: A Stablecoin for Cross-Chain Commerce }

\author{Dani\"el Reijsbergen}
\affiliation{%
  \institution{Nanyang Technological University}
 \country{Singapore}
}

\author{Bretislav Hajek}
\affiliation{%
  \institution{National University of Singapore}
 \country{Singapore}
}

\author{Tien Tuan Anh Dinh}
\affiliation{%
  \institution{Deakin University}
 \country{Australia}
}

\author{Jussi Keppo}
\affiliation{%
  \institution{National University of Singapore}
 \country{Singapore}
}

\author{Henry F.\ Korth}
\affiliation{%
  \institution{Lehigh University}
  \state{Pennsylvania}
 \country{USA}
}

\author{Anwitaman Datta}
\affiliation{%
  \institution{Nanyang Technological University}
 \country{Singapore}
}

\authorsaddresses{}

\newcolumntype{L}[1]{>{\raggedright\let\newline\\\arraybackslash\hspace{0pt}}m{#1}}
\newcolumntype{C}[1]{>{\centering\let\newline\\\arraybackslash\hspace{0pt}}m{#1}}
\newcolumntype{R}[1]{>{\raggedleft\let\newline\\\arraybackslash\hspace{0pt}}m{#1}}

\newcolumntype{Q}[2]{%
    >{\adjustbox{angle=#1,lap=\width-(#2)}\bgroup}%
    l%
    <{\egroup}%
}
\newcommand*\rot{\multicolumn{1}{Q{45}{1em}}}
\newcommand*\rott{\multicolumn{1}{Q{45}{1em}|}}

\newcommand{\rda}[1]{\textcolor{black}{#1}}

\newcommand{\rdc}[1]{\textcolor{black}{#1}}
\newcommand{\rdd}[1]{\textcolor{black}{#1}}

\newcommand{\confintv}[2]{\begin{tabular}{c} #1 \\[-0.125cm] \scalebox{0.7}{$\pm$#2} \end{tabular}}

\begin{document}

\begin{abstract}

Decentralized Finance (DeFi), in which digital assets are exchanged without trusted intermediaries, has grown rapidly in value in recent years. 
\textit{Stablecoins}, which are pegged to a non-volatile asset such as the US dollar, are a prominent feature of DeFi as they mitigate the risk associated with price fluctuations. 
However, existing stablecoin \rdc{systems} are tied to individual blockchain platforms, and trusted parties or complex protocols are needed to exchange stablecoin tokens between blockchains. 
Our goal is to design a practical stablecoin \rdc{system} for cross-chain commerce, \rdd{and to do so we must overcome} two main challenges. The first is to support a large and growing number of blockchains efficiently. The second
is for the stablecoin to be resilient to blockchain platform failures and to price fluctuations that affect its collateral. 
We present \name to address these challenges. We demonstrate \name{}'s efficiency by comparing \rdd{the performance of} a prototype implementation to related baselines, and its resilience through an empirical analysis
of historical token price data.


\end{abstract}

\maketitle

\section{Introduction}
\label{sec:introduction}
\textit{Decentralized Finance} (DeFi) promises to revolutionize finance by removing the need for trusted middlemen to
process transactions. In DeFi, transactions commonly involve a change of ownership of digital \textit{tokens} that exist
on a distributed ledger such as a \textit{blockchain}\footnote{In the remainder of this paper, we will use the terms `blockchain' or `chain' to refer to any type of distributed ledger, including ledgers that are technically not chains (e.g., Avalanche's directed acyclic graphs \cite{rocket2018snowflake}).} -- e.g., one user sending bitcoins to another
\cite{btc_origin}. Many blockchain platforms (e.g., Ethereum \cite{buterin2014next}) also support \textit{smart
contracts}, which enable automated versions of financial concepts such as auctions, loans, swaps, decentralized exchanges, and
ownership shares. The total value of the tokens held in DeFi smart contracts peaked at over \$150 billion (USD) in late 2021
\cite{werner2021sok,defillama}, with at least \$100 billion held in decentralized exchanges alone \cite{xu2023sok}, however it had fallen to around \$50 billion as of December 2023 \cite{defillama}. 

\begin{figure}
\centering
\subfloat[Fully Siloed]{\includegraphics[width=0.3\linewidth]{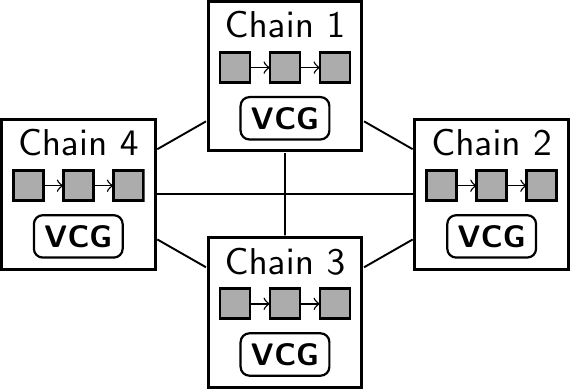}\label{fig:overview_siloed}}
\hspace{0.035\linewidth}
\subfloat[Wrapped + DAI]{\includegraphics[width=0.3\linewidth]{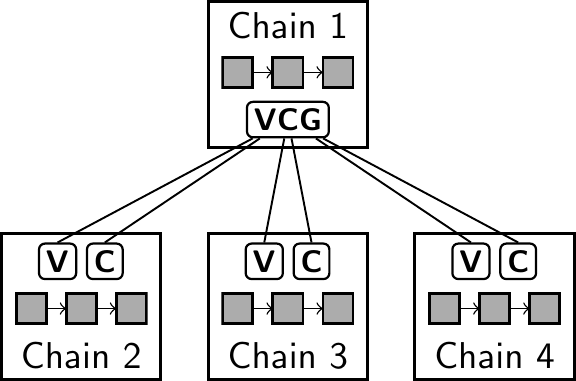}\label{fig:overview_dai_wrapped}} 
\hspace{0.035\linewidth}
\subfloat[\name]{\includegraphics[width=0.3\linewidth]{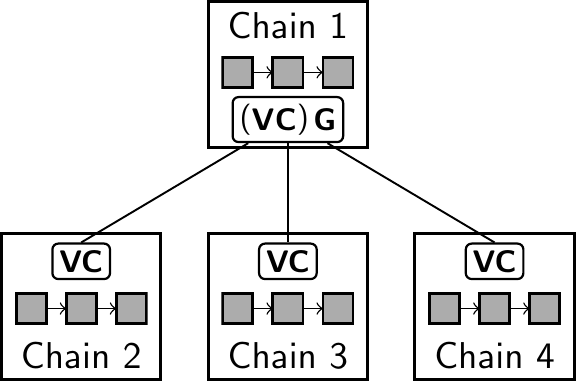}\label{fig:overview_crocodai}} 
\caption{Overview of different cross-chain stablecoin systems. \textbf{V}, \textbf{C}, and \textbf{G} respectively indicate the presence of vault, stablecoin, and governance contracts. }
\end{figure}

The DeFi ecosystem is currently fragmented over a multitude of different blockchain platforms which each offer their own choices in
terms of efficiency, security, functionality, and structure of DeFi contracts. For example, Cardano
\cite{hoskinson2017cardano,kiayias2017ouroboros}, Solana \cite{yakovenko2018solana}, Avalanche \cite{rocket2019scalable}, and Algorand \cite{gilad2017algorand} each had
a valuation of more than a billion USD in early 2023. These blockchains are not designed to \textit{interoperate}: i.e., there are no native
mechanisms for transferring tokens between chains.  However, \rda{the emergence} of blockchain
\textit{bridges}~{\cite{bridges,celer}} and dedicated interoperability platforms such as Cosmos {\cite{Polkadot20}} and Polkadot {\cite{cosmos20}} demonstrates an increasing
demand for cross-chain interoperability. 
A major challenge in cross-chain trades is \textit{atomicity}: the transactions in a trade must be committed on all blockchains, or not at all. 
Existing protocols for cross-chain atomicity, including \textit{relays} and \textit{hashed time locks}
\cite{herlihy2021cross,belchior2021survey,thyagarajan2022universal}, essentially lock the assets involved in a trade and only release (or
return) them once the transactions have been committed (or aborted) on all blockchains.  
While ensuring atomicity, this approach increases the financial risks for the users in the following ways: \begin{inparaenum}
\item the tokens held in escrow may lose value due to price fluctuations,
\item one party may abort the transaction, leading to an opportunity cost for the other parties, and
\item the DeFi platforms that facilitate the trade may be faulty or malicious\end{inparaenum}.
These risks are interrelated: if there is a significant change in the value of the tokens in escrow (1), then this may cause a party to abandon a trade (2), creating an opportunity cost to
the other users.
Price volatility can be mitigated in principle if the tokens
held in escrow are \textit{stablecoins} \cite{klages2020stablecoins,cao2021designing,moin2020sok} whose value is pegged
to a less volatile real-world asset (e.g., the US dollar). However, a flawed stablecoin design can itself put users'
tokens at risk, as witnessed by the collapse of the Terra/Luna stablecoin in May 2022 \cite{briola2023anatomy}, in which
stablecoins worth $\$$60 billion lost nearly all of their value. 

Our goal is to design a practical stablecoin as a central point of exchange for cross-chain commerce. 
We focus on stablecoin designs that are backed with crypto collateral -- e.g.,  Dai~\cite{daiwhitepaper} on the Ethereum blockchain -- to minimize the dependence on trusted parties outside the DeFi ecosystem (e.g., banks and asset management companies).
To realize our goal, we must address two main
challenges. The first challenge is to support a large and growing number of blockchains efficiently, including new and relatively low-use platforms. These low-use blockchains typically have lower transaction fees, but also less liquidity and existing DeFi support. The second challenge is to ensure that the stablecoin is resilient to \textit{black swan events} (e.g., 51\% attacks or critical smart contract bugs) that affect individual blockchain platforms, and to \textit{price fluctuations} that affect multiple blockchains simultaneously. If the total value of the collateral drops below the intended value of the issued stablecoins -- which we will refer to as \textit{system failure} -- then this may trigger the DeFi equivalent of a bank run, forcing stablecoin holders to sell their tokens at a loss. 
\rdc{Our key insight is that a wide range of collateral can address both challenges. In particular, it allows users from smaller chains to create robust stablecoins without interacting with a more expensive blockchain or relying on a trusted third party. Furthermore, it reduces the exposure of the stablecoin system as a whole to sudden crashes and price fluctuations.}

In this work, we present \name{}, a stablecoin for cross-chain commerce that addresses the above challenges. Compared to its closest baseline approaches, \name has greater efficiency, i.e., lower transaction costs, and better or similar resilience to black \rda{swan} events and price fluctuations.

A first baseline is for each blockchain to support its own crypto-backed stablecoin and  required components: \textit{vaults} to store collateral (\textbf{V}), \textit{coin} contracts to process stablecoin transfers (\textbf{C}), and \textit{governance} contracts for system-level decisions (\textbf{G}) -- this is depicted in \Cref{fig:overview_siloed}. \rda{In this baseline, stablecoins can be transferred between chains through bridges that are unique to each chain pair. Such bridges either lock tokens as collateral on one chain and create \textit{wrapped} \cite{wbtc} tokens on the other, or \textit{burn} tokens \cite{karantias2020proof} on one chain and mint them on the other.}\footnote{\rda{The Circle-created USDC bridge between Ethereum and Avalanche \cite{usdcbridge} is an example of the latter.}} However, the lack of liquidity on smaller blockchains makes it difficult to find collateral locally, and each individual stablecoin is highly vulnerable to black swan events and price fluctuations on its platform. 

A second baseline is to combine an existing crypto-backed stablecoin on a large (i.e., high-use) blockchain, e.g., Ethereum's Dai, with two bridges for each smaller chain: one to deposit wrapped versions of the smaller chains' native tokens on the large chain for collateral, and another to create wrapped versions of the stablecoin on the smaller chains. In this setting, each smaller chain has a vault contract (\textbf{V}) for collateral and an independent coin contract (\textbf{C}) for the wrapped stablecoins -- see \Cref{fig:overview_dai_wrapped}. However, this approach  requires unnecessary transactions on the large blockchain to create stablecoins or move them between chains -- such transactions are typically more expensive than on the smaller chains. By contrast, \name uses integrated vault and coin contracts (\textbf{V} and \textbf{C}) on each smaller chain, and leverages the governance layer on a dedicated relay chain for cross-chain transactions and system-level decisions, as depicted in \Cref{fig:overview_crocodai}. 

In this paper, we make the following contributions:

\begin{enumerate}
\item We present \name, a stablecoin for cross-chain commerce  (\Cref{sec:our_proposal}), and explore two main design variants: one uses a quorum of trusted nodes for the relay chain, and the other uses an existing blockchain (\Cref{sec:relay_chain_designs}). {While \name{} re-uses parts of the existing Dai stablecoin system, interactions with vaults are strictly local to the smaller chains, and it supports cross-chain price oracles to improve price discovery on smaller blockchains}.

\item We present a theoretical analysis of \name{} to show that its resilience to black swan events  (\Cref{sec:security_analysis}) and price fluctuations  (\Cref{sec:financial_analysis}) is greater than the first baseline (fully siloed), and similar to the second baseline (Dai + wrapped tokens), due to its greater ability to diversify its collateral. In particular, our analysis shows that \name{} has $\approx$$30\%$ lower risk of system failure than the first baseline, based on historical token price data from a period that includes the collapse of the FTX exchange. 

\item We present a prototype implementation of \name{} and perform microbenchmark and end-to-end experiments (\Cref{sec:implementation}) to compare \name{} to the second baseline (Dai + wrapped tokens) in terms of transaction costs and processing times. We find that the gas cost of stablecoin creation is only half as much in \name{} compared to the second baseline -- furthermore, \name{} can complete the entire process in a single block, whereas the baseline requires block finalization (which takes $32+$ blocks on Ethereum) on two chains for security. The total gas cost of cross-chain transfers is $\approx$$40\%$ higher in \name than in the baseline, but the baseline consumes $\approx$$120\%$ more gas on the (expensive) relay chain. \rdc{For example, if gas on the relay chain is $8\times$ more expensive than on the smaller chains,\footnote{\rdc{The $8\times$ difference is a conservative estimate based on a sample of gas cost differences between Ethereum and layer-2 chains such as Arbitrum and Optimism \cite{l2fees,ethereumgas,optimismgas} in early 2024. Recently, far greater differences ($200\times$ and above) have been observed, e.g., 0.06 GWei/gas for Optimism vs.\ 10 GWei/gas for Ethereum on 18 July 2024.}} then \name{} is roughly 40\% cheaper than the second baseline.}

\end{enumerate}
\rdc{In summary, \name{} is cheaper than a baseline with wrapped tokens, is less vulnerable to price fluctuations and black-swan events than a siloed stablecoin, and is not reliant on any trusted third parties unlike existing stablecoins such as Tether and Paypal USD. By facilitating stablecoin use on smaller chains and stablecoin transfers between chains, \name{} is a practical medium of exchange in cross-chain commerce \cite{datta2024blockchain}. To aid reproducibility, we have made the data and the source code for our experiments publicly available via} \url{https://github.com/ntublockchain/crocodai}.

The remainder of the paper is structured as follows.  \Cref{sec:background} provides background information on DeFi
and related work. 
\Cref{sec:our_proposal} presents the design of \name.
We analyze \name's resilience against black swan events in \Cref{sec:security_analysis} and price fluctuations in \Cref{sec:financial_analysis}.
\Cref{sec:implementation} details the prototype implementation and corresponding experimental results. \Cref{sec:conclusions} 
discusses future work and concludes the paper.

\section{Background}
\label{sec:background}

\subsection{General Background}

\hspace{0.5cm}\textbf{Tokens and DeFi. } 
In DeFi, the principal function of blockchain transactions is to transfer \textit{tokens}, which represent ownership of assets. There are two common types of tokens: \textit{fungible} tokens -- also commonly called \textit{cryptocurrencies} -- and \textit{non-fungible tokens} (NFTs) \cite{das2022understanding}. Fungible tokens of the same type are identical, whereas each NFT represents a unique asset.
Smart contracts can be used to define cryptocurrencies that are different from the blockchain's native cryptocurrency. In Ethereum, this is particularly facilitated by the \textit{ERC20 token standard} \cite{erc20} which defines a common set of smart contract functions required by any new cryptocurrency. 
Smart contracts can also support NFTs, e.g., via OpenSea or Mintable. 
A prominent use case for user-defined tokens is to create a version of one blockchain's token on another blockchain -- e.g., Wrapped Bitcoin (WBTC) on the Ethereum blockchain. In general, \textit{wrapped} tokens \cite{wbtc} are tokens maintained by \textit{custodians} -- e.g., for WBTC this includes the Kyber Network and BitGo -- who are trusted to create wrapped tokens if and only if an equivalent amount of tokens under their control have been locked on the supporting blockchain. \xclaim \cite{zamyatin2019xclaim} is an academic proposal that is similar to WBTC but which requires custodians to deposit additional collateral, thus adding a financial penalty as a disincentive for abuse (e.g., stealing the locked tokens).


\textbf{Stablecoins. } A stablecoin is a cryptocurrency whose value is pegged to an underlying fiat currency, e.g., the US dollar.
There are multiple types of stablecoins \cite{moin2020sok,klages2020stablecoins}, of which the most common ones are:
\begin{enumerate}[leftmargin=15pt]
\item \textit{Fiat-backed stablecoins}, which are backed by fiat currency held largely in cash or cash equivalents.
The essential guarantee that underpins the price stability of a fiat-backed stablecoin is that in case of system shutdown, each stablecoin can be used to redeem a corresponding amount of the pegged fiat currency.
\item \textit{Crypto-collateral-backed stablecoins} (or \textit{crypto-backed stablecoins} for brevity), which are backed by other cryptocurrencies that are stored in one or more smart contracts. In this case, the guarantee is that in case of system shutdown, each stablecoin can be redeemed for a basket of tokens whose value equals one unit of the pegged currency.
\item \textit{Algorithmic stablecoins}, which do not use reserves or collateral and which attempt to maintain a peg by adjusting the supply of a paired token \cite{moin2020sok}.
\end{enumerate}
A multitude of fiat-backed stablecoins exist, the most prominent (as of September 2023) being Tether (USDT), USD Coin (USDC), Binance USD (BUSD), and TrueUSD (TUSD). The most prominent crypto-backed stablecoin is DAI on the Ethereum blockchain, which is maintained by the MakerDAO foundation and which we discuss in more detail in \Cref{sec:minimal_dai}. \rdc{Other examples of crypto-backed stablecoins include Liquity USD (LUSD) and Reflexer’s
RAI, which are backed entirely by ETH. Other recent stablecoin proposals include FRAX \cite{kazemian2022frax}, which is an algorithmic stablecoin with a cryto-backed component, and Celo \cite{clabs2019celo}, which is backed by a mixture of fiat and crypto assets.}

During our observation period (see \Cref{sec:financial_analysis}), most fiat- and crypto-backed stablecoins did an overall sound job in maintaining their peg with the US dollar: the observed price ranges of DAI, USDT, USDC, BUSD, and TUSD were \mbox{[0.9921, 1.0062]}, \mbox{[0.9834, 1.0005]}, \mbox{[0.9991, 1.0073]}, \mbox{[0.9988, 1.0126]}, and \mbox{[0.9937, 1.0056]}, respectively.
The analysis in \cite{ma2023stablecoin} suggests that fiat-backed stablecoins have a trade-off between price stability and risk of system failure in the sense that better arbitrage improves the former at the cost of increasing the latter. \rdc{Recently, established financial services have launched their own fiat-backed stablecoins, e.g., PayPal's PYUSD, which is available on the Ethereum and Solana blockchains \cite{paypal}. Finally, there have been proposals by central banks to issue their own digital tokens, which would have the same functionality as stablecoins -- e.g., the digital euro, which recently completed its 2-year investigation phase and is now moving to the preparation phase \cite{euro}.}

Fiat-backed stablecoins introduce a trust assumption, namely on the agents who handle and store the fiat currency. In Tether, this trust assumption has been widely questioned -- accusations include Tether using fiat collateral to buy high-risk assets \cite{tetherallegation} and creating unbacked stablecoins to prop up the price of bitcoins \cite{griffin2020bitcoin}. Furthermore, algorithmic coins have nothing to underpin them in case of a \textit{death spiral} \cite{moin2020sok,briola2023anatomy}, as witnessed by the  loss of over 60 billion USD when the Terra/Luna stablecoin system failed. The main threat to crypto-backed stablecoins is that the system may fail if the value of the collateral decreases too rapidly -- \rda{however}, we show in this paper that this risk can be made small with sufficient overcollateralization. In the following, we focus particularly on Dai, which is the most prominent crypto-backed stablecoin. 

\subsection{The Dai Stablecoin System (DSS)}
\label{sec:minimal_dai}

In this section, we describe the existing single-chain DSS. Although there are documents and code repositories that describe the DSS in its various stages of evolution, e.g., \cite{daiwhitepaper,multicollateraldai,daigithub,moin2020sok,bhat2021simulating,chen2022understanding}, we focus on the GitHub code \cite{daigithub} and the description of the multi-collateral DSS \cite{multicollateraldai}. 

\textbf{Dai Tokens.}
In the DSS, Dai tokens are minted by users who deposit cryptocurrency tokens as collateral.
In the first iteration of the Dai stablecoin (the single-collateral Dai or ``Sai''), only ETH tokens could be deposited for this purpose.
The procedure for a user, Alice, who wants to mint new Sai tokens is as follows.
First, she must send a transaction to a Dai vault to create a \textit{Collateralized Debt Position} (CDP) \cite{daiwhitepaper}. 
Next, she must send ETH tokens as collateral to the contract to fund the CDP. 
Alice can use a funded CDP to create Dai stablecoins and withdraw them to her wallet, incurring \textit{debt}. At any later time, she can send Dai tokens back to the CDP to repay her debt. The above can be generalized to a multi-collateral (but single-chain) DSS by associating each CDP with a token type: in the DSS, this could theoretically include any ERC20 token, including wrapped tokens such as WBTC.
Debt increases exponentially over time, and the rate at which it increases is determined by the \textit{stability fee}.

\textbf{Liquidations.}
At all times, Dai CDPs must be \textit{overcollateralized}, i.e., the ratio of the total value of the Alice's collateral to her debt must exceed some threshold called the \textit{liquidation ratio} \cite{bhat2021simulating,daiwhitepaper}.
 By default, the liquidation ratio in Dai is set to $1.5$ -- i.e., Alice can withdraw Dai stablecoins worth at most $\frac{2}{3}$ of the total value of the collateral in their CDP \cite{moin2020sok}.
Similarly, Alice can return Dai and withdraw collateral from her CDP as long as the above condition holds. If a CDP has no associated debt, then it can be \textit{closed} and all remaining collateral can be withdrawn. If, at any time, a CDP's ratio of its collateral value to its debt is insufficient, then Alice's CDP is \textit{underwater} and can be liquidated. CDPs can go underwater either because the value of the collateral has dropped due to price fluctuations, or because the stability fee has caused the debt to become too high. To avoid liquidation, Alice must either deposit additional collateral, or raise Dai to clear an at-risk position -- if this fails, a \textit{collateral auction} can be started by other users. During a collateral auction, other users can bid to buy (portions of) Alice's collateral using Dai tokens. Once the CDP's debt has been cleared, 
Alice receives the remaining collateral (if any) minus a penalty. 

\textbf{Price Feeds.} To determine the value of deposited collateral, MakerDAO uses information from price feeds supplied by \textit{oracles} -- e.g., Chainlink or Uniswap contracts \cite{liu2021first}. The number of price feeds used by the DSS has increased from 14 between 2017 and November 2019 \cite{gu2020empirical} to 20 in late 2023 \cite{mips}.
Smart contracts that maintain CDPs -- i.e., the vaults -- periodically update token prices by setting each price equal to the median of a set of recent prices for the associated token.

\textbf{Governance.}\label{sec:dss_governance} All major decisions about the functioning of the protocol are taken by users who hold MKR, which is MakerDAO's governance token. Such decisions include 1) verifying and adding smart contracts that implement any of the functionalities above, 2) changing the Dai savings rate, and 3) occasionally changing other system-level parameters. The following are the main system-level parameters (see page 15 of \cite{multicollateraldai}):
\begin{inparaenum}
	\item the \textit{debt ceiling} for each token type, i.e., the maximum amount of total debt that can be created using a given collateral type,
	\item the stability fee,
	\item the liquidation ratio,
	\item the liquidation penalty,
	\item the duration of collateral auctions,
	\item the duration of auction bids, and
	\item the minimum auction step size (to prevent bids with a minor increment).
\end{inparaenum}

\textbf{Out-of-Scope Concepts.} The following DSS concepts are beyond the scope of this paper:
\begin{inparaenum}
\item the Dai Savings Rate mechanism, which allows users to deposit Dai to collect excess stablecoins created from the stability fee,
\item debt auctions, which occur when a collateral auction fails to raise enough Dai to cover the CDP's debt, 
\item surplus auctions, which occur when the amount of Dai left in vaults due to auctions and stability fees exceeds some threshold, and
\item emergency shutdowns, which can be triggered after a catastrophic event (e.g., system failure) and during which the total remaining collateral is distributed to Dai holders.
\end{inparaenum}

\subsection{Related Work}
\label{sec:related_work}

The DSS can be seen as an extension of \textit{red-black coins} \cite{salehi2021red}, which feature two token types of which one has reduced price volatility (e.g., Dai tokens vs.\ CDP debt). In \cite{salehi2021red}, single-asset price simulations are used to estimate the price of the volatile tokens -- the main difference with the analysis in \Cref{sec:financial_analysis} is that we estimate failure risks, that we use multiple correlated assets, and that we use a heavy-tailed distribution to better represent token price fluctuations. Price simulations of DSS collateral can also be found in \cite{bhat2021simulating}, and \cite{klagesmundt2022while} presents a stochastic model for a single-asset stablecoin. In \cite{grobys2021stability}, the relationship between stablecoin and Bitcoin price fluctuations is investigated. In \cite{ma2023stablecoin}, fiat-backed stablecoins are found to have a trade-off between price stability and risk of system failure.

For secure cross-chain transfers, \textit{cross-chain deals} \cite{herlihy2021cross} provide a modeling framework to describe sequences of cross-chain token transfers. 
In cross-chain deals, the cost of a participant abandoning the protocol before completion (\textit{sore-loser attacks} \cite{xue2021hedging}) can be made explicit -- this form of misbehavior can be disincentivized through a \textit{premium} \cite{engel2021failure} paid by the parties that abandon the protocol before termination.
Cosmos {\cite{Polkadot20}} and Polkadot {\cite{cosmos20}} are blockchain platforms that were specifically developed with blockchain interoperability in mind, but which do not trivially interoperate with existing blockchains outside their ecosystems. 
\rdc{However, a version of \name{} could be implemented as a native stablecoin on these platforms with collateral on multiple subchains. Alternatively, collateral on multiple chains can be integrated into \name{} through trusted relays or zero-knowledge proofs as discussed in \Cref{sec:our_proposal}.}
Finally, \xclaim \cite{zamyatin2019xclaim} claims to provide secure cross-chain transfers, but has limitations when transferring to/from smaller blockchain platforms that lack sufficient liquidity because the custodian must submit additional collateral.




\section{\name: A Cross-Chain Stablecoin}
\label{sec:our_proposal}

\begin{figure}[t]
\centering
	\includegraphics[width=0.6\textwidth]{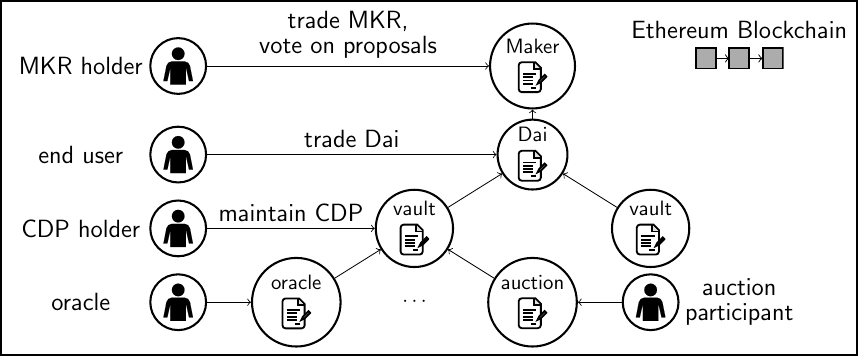}
\caption{Overview of the existing single-chain DSS. }
 \label{fig:single_chain_dss}
\end{figure}

In this section, we present \name, which extends the functionality of the single-chain DSS (depicted in \Cref{fig:single_chain_dss}) to a multi-chain setting (depicted in \Cref{fig:crocodai}).  Our main design goal is a stablecoin that is backed by collateral on multiple blockchains, and whose tokens can be moved securely and efficiently between chains. 
In the following, we will refer to blockchains that store collateral and stablecoin tokens as \textit{coin chains} (as per \cite{herlihy2021cross}). The communication between the coin chains is managed by a \textit{relay chain}, which may be an existing (coin) chain. The main difference between the single-chain DSS and \name is the inclusion of \textit{relay nodes} -- depicted with double circles in \Cref{fig:crocodai}. In \name{}, all vaults are local to the coin chains, and  cross-chain price oracles are supported to complement the oracles on smaller coin chains.
This allows \name{} to have lower costs for local operations such as creation and transfer of stablecoins on smaller blockchain platforms (we will discuss this in more detail in \Cref{sec:implementation}.)

In the following, we first discuss in more detail which DSS operations remain local and which require interaction with the relay chain. We then discuss the main components of the coin and relay chains. Next, we discuss two design variants for the relay chain: the first relies on a quorum of trusted nodes who perform most operations off-chain, and the second uses an existing blockchain. 
We conclude by listing the security requirements of a cross-chain stablecoin, which will be the subject of the Sections~\ref{sec:security_analysis}~and~\ref{sec:financial_analysis}.

\subsection{Overview}


\hspace{0.35cm}\textbf{Local Operations.}
\name{} supports the \rda{following} local operations.
\textbf{CDP Management:} Each coin chain has its own vaults that allow users to create CDPs using tokens on that blockchain. As such, stablecoins can be created locally on minor chains, i.e., without the need for transactions on a more expensive chain, as long as the debt limit on the minor chain is not violated. Since CDPs are local, liquidation auctions and other mechanisms such as savings rates are also processed locally.
\textbf{Local Stablecoin Transfers:} Each coin chain keeps track of the balance of stablecoins held by its accounts, so the stablecoins can be exchanged locally between accounts. 
\textbf{Local Price Information:} 
Finally, local oracles may provide the vaults with price information.

\textbf{Cross-Chain Operations.} The following operations require relay chain involvement. 
\textbf{Cross-chain Transactions:} In a cross-chain transaction, stablecoins are burned on one coin chain and an equivalent number is minted on another coin chain, as we discuss in more detail below. 
\textbf{Ward Management:} Governance layer approval is needed to {add/remove addresses} that are authorized to call security-critical functions such as mint or burn on another contract -- we call such addresses \textit{wards} as per DSS terminology.
Periodically, \rda{the governance layer may grant or revoke ward status to/from smart contracts} for a variety of reasons -- e.g., 
\begin{inparaenum}
\item to add support for an emerging blockchain platform, 
\item because a critical bug has been found in an existing smart contract, 
\item because a chain has undergone a hard fork that interferes with the functioning of the contract, or
\item to implement a more efficient version of a component.
\end{inparaenum}
\textbf{System Parameter Updates:} Governance layer approval is also needed to update the system parameters discussed in \Cref{sec:dss_governance}, e.g., a chain's debt ceiling or a token's overcollateralization rate on a chain.
\textbf{Global Price Information:} The final function is to provide an additional source of token price information to vaults through relay chain oracles --
this particularly benefits smaller blockchains that lack a mature DeFi ecosystem.

The exact mechanism for cross-chain transfers is as follows: a user who holds stablecoins on a coin chain
first submits a cross-chain transfer request to a relay node and transfers the tokens for escrow. The request is added to a hashmap where it awaits commitment. The nodes on the relay chain monitor the coin chains, and once they detect a cross-chain transfer request, they sign a commit message. When the coin chain receives the commit message, it commits the transaction and burn the tokens held in escrow. A commit message is also sent to the relay node on the target blockchain, which  mints an equivalent number of stablecoins. If the transaction is aborted for any reason, then the stablecoins in escrow are returned to the sender.

\begin{figure*}[t]
	\centering
	\includegraphics[width=\linewidth]{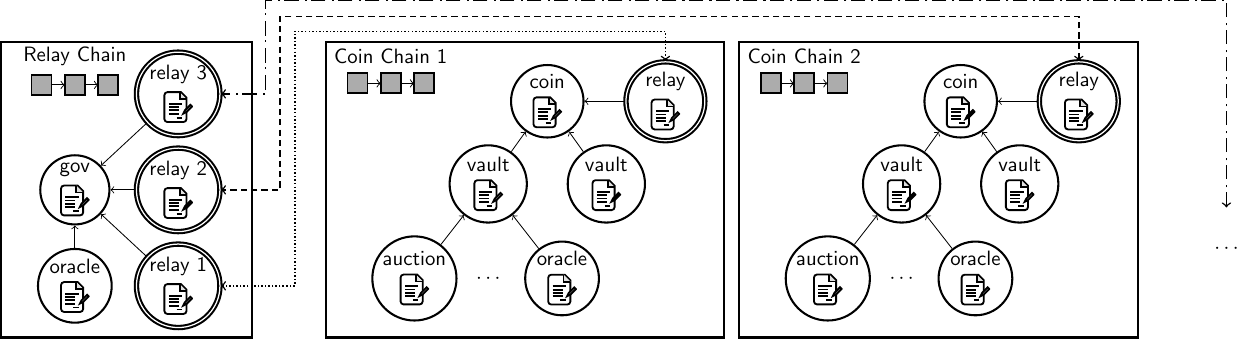}
	\caption{Overviews of \name, a multi-chain stablecoin system based on the DSS. The relays (double circle) are not shared with the single-chain DSS.
 }
  \label{fig:crocodai}
\end{figure*}

\subsection{Core Components}

\hspace{0.35cm}\textbf{Coin Chain Components.}
Each coin chain has the following components.
\textbf{Coin:} implements a contract that is similar to \texttt{Dai} in the DSS, and processes transfers of stablecoins between users on the coin chain. There is only one coin contract per chain.
\textbf{Vault:} implements contracts that are similar to \texttt{Vat}, \texttt{GemJoin} and \texttt{DaiJoin} in the DSS,  which allow users to create CDPs and withdraw/deposit collateral and stablecoins. There can be multiple vaults per chain.
\textbf{Auction:} implements auction contracts that liquidate the collateral of underwater CDPs, with any appropriate auction logic -- e.g., \texttt{Cat} or \texttt{Dog} in the DSS. There can be multiple auction contracts per chain.
\textbf{Bottom-Level Oracle:} implements a price feed using a contract similar to \texttt{Spot} in the DSS, combined with a Medianizer.
\textbf{{Bottom-Level} Relay:}  validates messages from the relay chain. It processes 
\begin{inparaenum}
    \item the addition of wards to \name{} contracts,
    \item incoming cross-chain transfers,
    \item commitments of outgoing cross-chain transfers, and
    \item updates of system parameters.
\end{inparaenum}

\textbf{Relay Chain Components.}
The relay chain has the following components.
\textbf{Governance:} tracks the ownership of governance tokens, and allows users to transfer them. \rdc{As in the DSS,} users who hold governance tokens can vote on system-level actions such as adding \rda{wards and} updating system-level parameters. \rdc{Our security assumption regarding governance token holders is that a majority is honest and online at any point during the voting window for system-level decisions.}
\textbf{Top-Level Oracle:} The relay chain may implement system-level price feeds to complement the price feeds on the coin chains.
\textbf{{Top-Level} Relay:} To process cross-chain transfers, the relay chain must be able to process information from the coin chains, particularly the transfer requests -- we discuss this in more detail below.

\subsection{Relay Chain Designs}
\label{sec:relay_chain_designs}

In this section, we discuss two possible design variants of the relay chain. The first relies on a quorum of trusted nodes with known public keys -- each trusted node runs an off-chain (full or light) blockchain client for each of the underlying chains to check the validity of cross-chain transfers, and the nodes maintain a distributed ledger for governance decisions. The second uses smart contracts on an existing chain to implement the consensus logic of the underlying chains (potentially using zero-knowledge cryptography) to verify cross-chain operations on-chain.

\textbf{Off-Chain Blockchain Clients.} In this approach, the stablecoin system is maintained by $n$ relay nodes with known public keys. The relay nodes use their (light) clients to detect and validate cross-chain transactions, and generate signatures using their private/public key pair to commit cross-chain transfers and \rda{process system-level actions approved by governance token holders}. In this setting, the only functionality of the bottom-chain relay contracts is to verify digital signatures. 
We discuss three different subvariants of this approach in \Cref{sec:implementation}.
\rdc{As is common in Byzantine fault tolerant systems \cite{castro1999practical, lamport1982byzantine}, we assume that at most $f$ nodes are faulty or malicious such that $n \geq 3f+1$.}
We note that it is vital that the relay nodes are controlled by independent entities: e.g., in the attack against the Ronin bridge, 5 out of 9 validators were successfully attacked \cite{roninhack}, of which 4 belonged to the same entity \cite{roninhack2}.

\textbf{On-Chain Consensus Logic.} In this approach,  the consensus logic of the coin chains is implemented in the {top-level relay contracts}, and the consensus logic of the relay chain is implemented in the {bottom-level relay contracts}.
 Essentially, the functionality of a light client -- i.e., validating block headers and proving the state of a contract using the header -- is implemented inside the smart contracts. This is the approach taken in, e.g., \xclaim \cite{zamyatin2019xclaim}. For example, in Bitcoin's Nakamoto Consensus \cite{btc_origin} this means that the smart contract computes header hashes and compares their integer value to the proof-of-work difficulty,\footnote{\rda{This may be computationally efficient if blockchains share built-in hash functions, otherwise this can be expensive.}} whereas in Proof-of-Stake Ethereum the smart contract verifies that the total stake of the nodes who have approved a block is sufficient. \rda{In this setting, the main role of the relay nodes is to generate and submit validity proofs.}
 Techniques to increase the efficiency of light clients have been proposed in the literature for both proof-of-work (e.g., Bitcoin \cite{daveas2020gas}) and proof-of-stake blockchains  (e.g., Algorand \cite{micali2021compact} and Mina \cite{bonneau2020mina}) \rda{-- in some cases using zero-knowledge proofs that may incorporate the blockchain's consensus logic \cite{xie2022zkbridge}.} 
 \rdc{Zero-knowledge proofs that are verifiable on multiple chains, e.g., ZKM's entangled rollups \cite{zkm2024entangled}, can also be used to bridge the chains that implement collateral to a layer-2 network that implements the vaults.}
 Finally, the data on the different coin chains can be integrated using TrustBoost \cite{wang2022trustboost}.
 
Although the second design variant avoids a trust assumption on a quorum of relay nodes, current iterations of such designs have a large computational overhead, typically requiring several kilobytes of data sent to relay contracts and more than 10 seconds to generate signatures \cite{micali2021compact,xie2022zkbridge}.\footnote{Improving the performance of such zero-knowledge schemes, e.g., using folding schemes  \cite{kothapalli2022nova}, is an active research area.} As such, we focus on the first variant in \Cref{sec:implementation} and leave the second as future work.

\subsection{System Requirements}
\label{sec:requirements}

The central requirement that we focus on in this paper is the following:
\begin{enumerate}
\item[(*)] \rdc{\textbf{Sufficient Overcollateralization.}} At any time $t \geq 0$, the ratio of the value of the collateral tokens to the total number of stablecoins in the system exceeds a \rda{\textit{safety threshold}} $\theta > 1$.
\end{enumerate}
The motivation behind this requirement is as follows: if the value of the collateral becomes too low, then the system is no longer able to reimburse all stablecoin holders with an equivalently valued basket of tokens in case of system shutdown. As such, we can expect users to start selling their stablecoin tokens at a loss to redeem as much value as possible, causing the stablecoin to lose its peg. \rdc{We assume that this sell-off may occur before the critical point is reached at which the value of issued stablecoins exceeds that of the collateral, because nervous investors may preempt the loss of the peg. As such, we choose a value of $\theta$ that is greater than $1$: a value of strictly $1$ corresponds to a stablecoin system that fails as soon as the aforementioned critical point is reached, whereas $\theta>1$ also covers scenarios in which system failure occurs earlier, in anticipation of the depegging.}

Due to random price fluctuations, we cannot guarantee with absolute certainty that full backing will hold in all circumstances. However, we use our analysis to show that for reasonable choices of $\theta$ and the liquidation ratio, the probability that full backing is violated after a medium-length period (i.e., one day) is very small for historical data. 
In \Cref{sec:security_analysis}, we also investigate under which conditions full backing continues to hold in the event of a black swan event such as a sudden token price collapse or a 51\% attack on a blockchain. 

We also require the following properties under the threat and token price model.
\begin{enumerate}[leftmargin=15pt]

\item \textbf{Safety.} Users do not lose tokens during stablecoin transfers and CDP operations.

\item \textbf{Liveness.} Stablecoin transfers and CDP operations eventually happen.

\item \textbf{Safety of Collateral} (informal). The probability that a minority of corrupted oracles cause a CDP that is not underwater to be inadvertently liquidated is small.
\end{enumerate}
More formally, for the last property we consider a minority of corrupted price oracles whose goal is to convince the protocol to accept a price $p^*$ instead of a true base price $p$ -- safety of collateral then holds if their probability of success vanishes exponentially as $|p-p^*|$ increases (see also \Cref{thm:corrupted_price_feeds} in \Cref{sec:security_analysis}).

As we will see in \Cref{sec:security_analysis}, the former two properties largely follow from results in the existing scientific literature. The latter is a result that we prove for the multi-chain setting in \Cref{sec:security_analysis}, but which is also applicable in the single-chain setting.

\subsection{\rdc{Illustrative Application}}
\label{sec:end_to_end}

\rdc{
We consider an example of cross-chain commerce: a setting in which an asset -- e.g., a basket of cryptocurrencies or an NFT -- is auctioned on one chain, and for which there are interested bidders who own tokens on multiple other chains. For the duration of the auction, price stability is important because changes in the value of bids could lead to unnecessary resubmissions or withdrawals of bids, causing a diminished user experience and/or opportunity costs \cite{xue2021hedging}. The auction example allows us both to illustrate the end-to-end processes of creating and transferring stablecoins in \name{}, as well as to compare \name{} to existing approaches in terms of performance and user experience. For this example, we will refer to the chain on which the asset is sold as the `auction chain' and to the chains on which bidders hold tokens as `coin chains'.}

\rdc{To achieve an auction with price stability, we consider two main approaches:
1.\ bidders use a centralized exchange (CEX) to swap their tokens into stablecoins on the auction chain, 2.\ bidders using \name{}'s built-in mechanism for cross-chain transfers. The first represents the most straightforward way to achieve a cross-chain auction in the current ecosystem by leveraging established exchanges such as Coinbase or Binance, whereas the second is our proposal. We also discuss other approaches at the end of this section.
}

\begin{figure*}[t]
	\centering
    \subfloat[][CEX Baseline]{\includegraphics[width=0.45\linewidth]{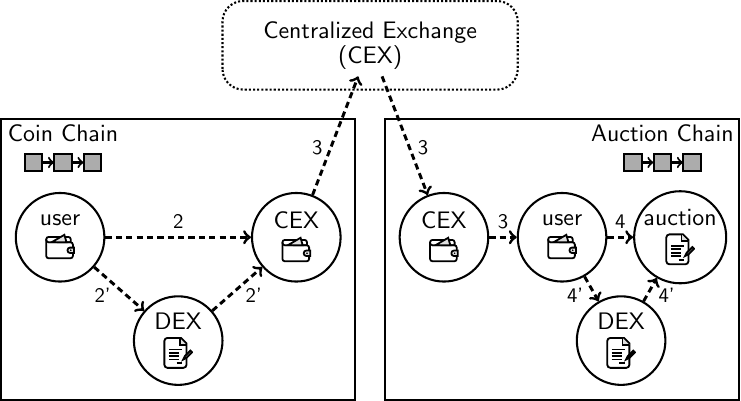}\label{fig:end_to_end_cex}}\hspace{0.04\linewidth}
    \subfloat[][\name{}]{\includegraphics[width=0.45\linewidth]{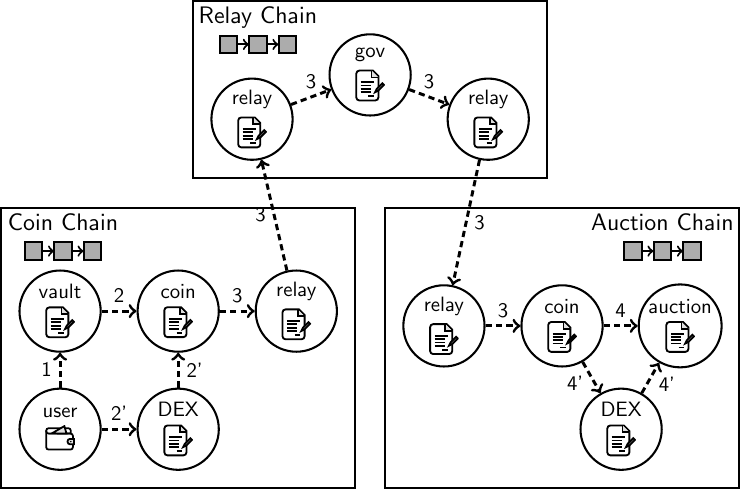}\label{fig:end_to_end_crocodai}}
	\caption{\rdc{Steps of a cross-chain exchange in the CEX baseline and \name{}. The edge labeling corresponds to the step numbers in \Cref{sec:end_to_end}.
    }}
  \label{fig:end_to_end}
\end{figure*}

\rdc{\textbf{Comparison.}}
\rdc{In the first approach, we assume the existence of a centralized exchange that can be used to buy and sell stablecoins on the coin chain and auction chain. In particular, the bidder would then complete the following steps: 
\begin{enumerate}
    \item Create an account with the CEX if she doesn't have one already.
    \item Transfer her tokens to the CEX to add it to her wallet on the CEX.
    \item Sell the tokens in her account and buy stablecoins, or swap the tokens for stablecoins directly if supported by the CEX, and withdraw the stablecoins to her account on the auction chain.
    \item Place a bid by transferring the stablecoins to the auction contract.
\end{enumerate}
In \name{}, the bidder would complete the following steps: 
\begin{enumerate}
    \item Deposit her tokens into a \name{} vault.
    \item Create new \name{} stablecoins and withdraw them to her account.
    \item Issue request to the relay chain for a cross-chain transfer of her \name tokens on the coin chain to an account owned by her on the auction chain.
    \item Place a bid by transferring the stablecoins to the auction contract.
\end{enumerate}
The fundamental difference between the two approaches is the reliance on a trusted third party in the first approach. By contrast, the second approach relies on a decentralized relay chain. Furthermore, the CEX baseline is more likely to suffer from a lack of liquidity: whereas steps 1 and 2 in the first baseline always rely on the CEX as a counterparty to swap tokens, in \name{} the bidder can always create new stablecoins and transfer them to the auction chain via the relay chain. In terms of fees and transaction delays, there are no reasons to expect major differences between the two approaches, since the cost of sending and processing cross-chain messages is typically small compared to the cost of monitoring the supported coin chains through a blockchain node, and the latter is a requirement for all cross-chain systems.}

\rdc{\textbf{Additional Intermediaries.}
In both approaches discussed above, either the CEX or \name{}'s relay chain acts an an intermediary to facilitate the cross-chain token transfer. Additional steps that involve intermediaries may be needed if 1) the CEX does not support sales of the bidder's tokens, 2) \name{} does not accept the bidder's tokens as collateral (e.g., because they belong to an unsupported type, or because the token type's debt limit would be violated), 3) the auction contract does not accept any token sold by the CEX, or 4) the auction contract does not accept \name{} stablecoins. In those cases, a DEX may exists that allows for on-chain swaps between unsupported and supported tokens. The additional steps are involving DEX swaps are depicted in \Cref{fig:end_to_end} through step $2'$, which replaces step 2, and step $4'$, which replaces step $4$.}



\rdc{\textbf{Other Approaches.}}
\rdc{A cross-chain auction can also be implemented using the framework of cross-chain deals \cite{herlihy2021cross}. However, this requires that the auctioneer creates smart contracts on each of the supported coin chains, instead of the just the auction chain. Finally, the bidder could create a wrapped version of her tokens on the coin chain, but this would typically require another trusted intermediary and is therefore similar to the CEX baseline.}




\subsection{\rdc{Deployment}}

\rdc{Real-world use of \name{} would require the deployment of the relay chain, the smart contracts (\rdd{e.g., coin and vault contracts}) on the supported smaller chains, and the development of tools that allow users to interact with the system. The costs of developing and deploying the relay chain are similar to those of existing bridges such as Celer\footnote{\url{https://cbridge.celer.network/bridge/ethereum-polygon/eth/}} and Relay\footnote{\url{https://relay.link/bridge/ethereum}}. The deployment of smart contracts has roughly the same cost as in a fully siloed system, and are small over long run compared to fees (e.g., at a gas price of 20GWei and an ETH price of \$3000, deploying a 1m gas contract would cost \$60). The throughput and scalability of the relay chain are highly dependent on the design choices and  implementation, e.g.,   Celer's cBridge claims a latency of at most 20 minutes during peak demand periods \cite{cbridge}, whereas the zkBridge prototype has a latency of under 2 minutes \cite{xie2022zkbridge} but at a high computational cost, while the Hephaestus protocol \cite{belchior2023hephaestus} claims to be able to process 600 cross-chain transactions in under 6 seconds.}

\rdc{In practice, we expect the user interface design to be very similar to existing designs for the DSS: web apps such as summer.fi\footnote{ \url{https://summer.fi/}} allow users to interact with vaults and to process local and cross-chain transfers. The incentives for developers are the same as for Dai, i.e., a combination of fees, publicity (an app can draw attention to other services), and in some cases volunteer work.
The active DAI ecosystem, including summer.fi, daistats.com \cite{daistats}, DeFiExplore,\footnote{\url{https://maker.defiexplore.com}} demonstrates that such incentives are sufficient in practice. }

\rdc{Finally, the system parameters in \name{} can be set to similar values as in the DSS. The overcollateralization rates, i.e., $\gamma$, will likely be in the order of $30-50\%$ for the supported token types, as they are in existing Dai vaults \cite{daistats}. Similarly, the debt limits for smaller blockchains are expected to be similar to those for ERC20 tokens such as Aave and Gnosis -- as in the DSS, these are proposed and voted on by governance token holders.}

\section{Security Analysis}
\label{sec:security_analysis}
In this section, we analyze whether our system satisfies the following technical properties from \Cref{sec:requirements}: \textit{safety} and \textit{liveness}, \textit{full backing} under the threat model, and \textit{safety of collateral}. We first discuss our threat model and then present our results. In particular, we prove that if tokens whose total debt limit\footnote{We note that price crashes in different tokens and platforms may occur simultaneously, e.g., the collapse of FTX also had consequences for other platforms such as Solana.} is $\zeta'$ suffer a price crash (\Cref{sec:token_crash}) or platform-level attack (\Cref{sec:compromised_coin_chain}), then full backing continues to
hold if $\theta < \gamma(1-\zeta')$, where $\gamma$ denotes the liquidation ratio (see \Cref{sec:minimal_dai}). In practice, this allows stablecoin holders to judge whether they find the risk of system failure low enough given their estimations of $\theta$ and $\zeta'$. We conclude by comparing the security of \name{} to the two baselines from the introduction. The technical proofs of the main theorems can be found in \Cref{sec:proofs}.

\rda{The following notation is used throughout this section and \Cref{sec:financial_analysis}. Token prices are represented in terms of their exchange rate with a non-volatile base asset, e.g., the US dollar (the \textit{numéraire}). Let $M$ denote the set of all tokens supported by \name{}. We denote the base price of token $m \in M$ at time $t > 0$ by $p_{m\,t}$, and the initial base price of token $m$ by $p_{m\,0}$. (Here, $t=0$ represents an arbitrary start point of time.) Given a batch of $k \geq 0$ tokens of type $m$, the \textit{value} of this batch at time $t$ is given by $k p_{m\,t}$. We assume that the price of token $m$ at time $t$ returned by an honest price oracle $o \in [O]$ equals $p_{m\,o\,t} \triangleq p_{m\,t} + \Delta_{m\,o\,t}$, where $\Delta_{m\,o\,t}$ is a random variable with a normal/Gaussian distribution \cite{ross2011elementary} with mean $\mu_{m\,o} \in \mathbb{R}$ and variance $\sigma^2_{m\,o} > 0$. The assumption of normality is in accordance with \Cref{sec:threat_model}, and is motivated by the common assumption that measurement errors have a normal distribution, e.g., in linear regression \cite{dobson2018introduction}.}

\subsection{\rdd{Security} Model}
\label{sec:threat_model}
\rdd{In the following, we first discuss the threat assumptions that are specific to \name{}, and then discuss the threats and countermeasures shared with generic stablecoin and cross-chain systems.}

In \rdd{\name's security} model, a majority of the governance token holders are honest and online, and the relay chain satisfies safety and liveness. Each underlying blockchain is either compromised or not. If it is not compromised, then we assume that it guarantees safety and liveness for all of its transactions. However, both safety and liveness can be violated in a compromised blockchain -- in particular, transactions in which collateral is sent to/from a CDP can be rolled back. We assume that the underlying communication networks satisfy partial synchrony, i.e., bounded delivery times, \rdd{to ensure that nodes do not wait indefinitely for information to be transmitted between chains. It has shown that liveness cannot be guaranteed in the presence of even a single faulty node in a fully asynchronous model \cite{fischer1985impossibility}. We do note that \name{} itself provides synchrony to the overall cross-chain process by requiring that a transfer is completed on the sending chain before processing it on the receiving chain, as depicted in \Cref{fig:end_to_end_crocodai}.}
For the price oracles, a majority is online and honest at all times, in which case they report a measurement error that is random and normally distributed. \rda{Furthermore, the output of each oracle is independent of the output of other oracles.} Dishonest oracles can report arbitrary prices. \rdd{We emphasize that these assumptions are common in the scientific literature, e.g., Byzantine fault tolerant protocols also require honest and online (super)majorities and partial synchrony \cite{lamport1982byzantine}. }

\rdd{As a cross-chain stablecoin system, \name also inherits threats that affect generic blockchains, cross-chain infrastructures, and stablecoins. Recent surveys that investigate the security of such systems are available for individual blockchains \cite{homoliak2020security,zhang2019security}, cross-chain systems \cite{belchior2021survey,zamyatin2021sok,herlihy2021cross}, and stablecoins \cite{mell2023understanding}. Threats that affect individual blockchains include the following categories:}
\begin{itemize}
    \item \rdd{\textit{Network-layer threats}, which may prevent users from connecting with honest blockchain nodes, potentially violating safety (if the remaining nodes are malicious) or liveness. Examples include eclipse attacks on blockchain peers, DNS/routing attacks, and Denial-of-Service (DoS) attacks. Countermeasures include generic network security proposals such as DNSSEC and BGPsec, or blockchain-specific defences such as SABRE \cite{apostolaki2019sabre}.} 
    
    \rdd{Consensus protocols that announce a single block proposer beforehand -- e.g., many Proof-of-Stake protocols \cite{brown2019formal} -- may be vulnerable to targeted DoS attacks against the block proposer, violating liveness. \name{} itself does not introduce such a vulnerability because cross-chain transfers can be initiated by any of the relay nodes.}
    
    \item \rdd{\textit{Consensus-layer threats}, which may cause established transactions to be reverted or the network to accept incorrectly formed blocks. Examples include majority (51\%) attacks or selfish mining. Countermeasures include protocol-level defences, e.g., finality/checkpointing protocols, and robust incentive mechanism design.}
    
    \rdd{Finality protocols are also helpful to \name{} implementations, because \name{} needs a logic for each supported chain to determine when a cross-chain transfer has been processed (i.e., highly unlikely to be reverted). Finality protocols provide such a logic by default.}
    
    \item \rdd{\textit{Contract-layer threats}, which may allow attackers or malicious owners to subvert the intended functionality of a smart contract. Examples include smart contract bugs or deliberate backdoors. Countermeasures include static code analysis or formal verification, potentially run by third parties in the case of malicious owners or insider threats. \name{}'s similarity to the DSS suggests that we can leverage a larger market for smart contract validation services.}
    
    \item \rdd{\textit{Application-layer threats}, e.g., bugs or backdoors in applications that support the stablecoin system, such as tools that allow users to interact with CDPs. As with contract-layer threats, countermeasures include static code analysis and formal verification.}
\end{itemize}
\rdd{The main threats to stablecoin systems as enumerated by researchers at NIST \cite{mell2023understanding} are categorized as security, stability, and trust ``issues''. The following security ``issues'' are mentioned: \begin{inparaenum}
\item Unauthorized or Arbitrary Minting of Stablecoins,
\item Collateral Theft,
\item Malicious Smart Contract Update and Hijack,
\item Data Oracles,
\item Exploiting the Underlying Blockchain, and
\item Writing Secure Software and Vulnerabilities.
\end{inparaenum}
Of these threats, (1), (2), and (5) are covered by the analysis in Section~\ref{sec:compromised_coin_chain}, as all scenarios effectively result in threats to the system's overcollaterization. Finally, (4) is covered by the analysis in Section~\ref{sec:compromised_price_feeds}, whereas (3) and (6) are generic blockchain threats as discussed previously, which can be addressed using the countermeasures mentioned in \cite{homoliak2020security,zhang2019security}.}

\rdd{Security threats to generic cross-chain systems are an active area of research \cite{belchior2021survey,zamyatin2021sok}, although some threats such as cross-chain replay attacks and data availability violations have been studied in the scientific literature. We discuss the fundamental properties of safety and liveness in \name{} in more detail below.
}

\subsection{Technical Safety and Liveness} 
For local operations, the consensus protocol of the coin chain on which the local transfer occurs ensures the safety and liveness of operations on this chain -- see, e.g., \cite
{garay2015bitcoin} \rda{ -- under the assumption of partial synchrony}. This includes local stablecoin transfers, CDP operations, auctions, and interest withdrawals. As such, tokens cannot be stolen during local transfers or CDP operations or locked forever in a CDP unless the coin chain is compromised. For cross-chain operations, safety is violated if tokens in a bottom-level relay are burnt without a corresponding amount being minted on the receiving chain, and liveness is violated if tokens are held in escrow  forever. However, the relay chain satisfies safety and liveness in our threat model, i.e., the transfer is eventually committed and coins are burned and minted on the outgoing and incoming chains, respectively. An alternative way to reach this conclusion is to view cross-chain transfers in \name as simple cross-chain deals \cite{herlihy2021cross} with the sender as the only party who votes, and the relay chain as a certified blockchain -- depending on the design (as discussed in \Cref{sec:relay_chain_designs}), either $2f+1$ nodes or the smart contracts on the relay chain sign the commit vote. 


To ensure safety and liveness, we emphasize that it is essential that the relay network is not compromised. A compromised relay network can burn coins without creating a corresponding amount on the receiving chain, violating user safety. Likewise, if the relay network is unable to process requests because more than $f$ relay nodes are offline or malicious, then the system does not satisfy liveness as cross-chain transfer requests may never happen. We also note that if the relay network is compromised, then full backing can be violated by creating an incoming transfer request without a matching outgoing transfer request, thus creating tokens without collateral. Although the security of the \textit{relay} chain is therefore essential to \name, we investigate the impact of compromised \textit{coin} chains on full backing below.


\subsection{Sudden Crash of Tokens}
\label{sec:token_crash}
Occasionally, the price of a cryptocurrency token may suddenly collapse, e.g., because of the discovery of a serious bug in one of its smart contracts, or because of off-chain events such as a bankruptcy or a legal inquiry into the organization responsible for its maintenance. As an example: the bankruptcy of Celsius caused the value of the CEL token to collapse, whereas the inquiry into FTX did the same for the FTT token and Serum (SRM). In this case, the collateral auctions are likely to fail to generate enough revenue and the system's over-collateralization is under threat. In the following, we consider a worst-case scenario: the price of some tokens falling to zero within a single time slot, i.e., between a given time slot $t'$ and time slot $t''$ -- for simplicity, we assume $t'=0$ and $t''=1$. We show that \name remains fully backed at time $t''$ if it is sufficiently over-collateralized and if the total debt created using tokens on the affected chains is limited.

Formally, \rda{let $\mathcal{J}_t$ denote the set of all CDPs at time $t$, $c_{j\,t}$ the total number of collateral tokens held in CDP $j \in \mathcal{J}_t$ at time $t$, and $D_{j\,t}$ the outstanding debt resulting from Dai withdrawals/deposits from/to CDP $j \in J_{t}$ at time $t$. We denote the token type associated with CDP $j$ by $y_j \in M$.
Let $C_t^* = \sum_{j \in J} p_{y_j \, t} c_{j\,t}$ denote the total value of the system's collateral at time $t$ and $D_t^* = \sum_{j \in J} D_{j\,t}$ the total debt. 
Let $M' \subset M$ denote the set of affected tokens} and \mbox{$\zeta' = \sum_{m \in M'} \zeta_m$} the total debt ceiling of the affected tokens. Let 
\begin{equation*}
C^{-}_t = \sum_{\substack{j \in \mathcal{J}_t \\ y_j \in M'}} p_{y_j\,t} c_{j\,t}, \;\;\;\;\; C^{+}_t = \sum_{\substack{j \in \mathcal{J}_t \\ y_j \notin M'}} p_{y_j\,t} c_{j\,t}
\end{equation*}
 i.e., $C^{-}_t$ represents the total value of the collateral tokens affected by the collapse, and $C^{+}_t$ the total value of the unaffected tokens. Similarly, let
\begin{equation*}
D^{-}_t = \sum_{\substack{j \in \mathcal{J}_t \\ y_j \in M'}} D_{j\,t}, \;\;\;\;\; D^{+}_t = \sum_{\substack{j \in \mathcal{J}_t \\ y_j \notin M'}} D_{j\,t}
\end{equation*}
 so that $D^{-}_t$ represents the debt associated with affected CDPs and $D^{+}_t$ the debt of unaffected CDPs at time $t$.
We consider an attack that occurs at time $1$, so that $C^-_{1} = 0$, $C^+_{1} = C^+_{0}$, $D^-_{1} = D^-_{0}$, and $D^+_{1} = D^+_{0}$. We can then prove the following theorem (see \Cref{sec:proofs}):

\begin{theorem}
Let $C^-_{1} = 0$, $C^+_{1} = C^+_{0}$, $D^-_{1} = D^-_{0}$, and $D^+_{1} = D^+_{0}$. If $p_{y_j\,0} c_{j\,0}/y_{j\,0} > \gamma$ for all $j \in \mathcal{J}_0$, then 
$C^*_{1}/D^*_{1} > \gamma(1 - \zeta')$, i.e., full backing holds at time $t=1$ if $\theta < \gamma(1 - \zeta')$.
\label{thm:sudden_price_drop}
\end{theorem}

\subsection{Individual Compromised Coin Chain} 
\label{sec:compromised_coin_chain}
If a majority of consensus power on a coin chain is compromised then the consensus nodes could perform a 51\% attack in the following way: send some amount of stablecoins to another chain, wait for the transaction to be committed by the relay network, then fork the compromised coin chain and undo the outgoing transfer request and the burning of the blockchain tokens. Similarly, they could undo the funding of a CDP. As a consequence, stablecoins are created on another chain without backing collateral. In the worst case, this attack would be performed simultaneously on all of the token types supported by the compromised chain. Again, let $M' \subset M$ denote the set of affected tokens. We again consider a worst-case attack that has been completed at $t=1$. Let $h'$ be the total amount of stablecoins created in the attack, so that 
$
D^{-}_{1} = D^{-}_{0} + h'.
$ 
We can then prove the following theorem (see \Cref{sec:proofs}):
\begin{theorem}
Let $C^-_{1} = C^-_{0}$, $C^+_{1} = C^+_{0}$, $D^-_{1} = D^{-}_{0} + h'$, and $D^+_{1} = D^+_{0}$. If $p_{y_j\,0} c_{j\,0}/y_{j\,0} > \gamma$ for all $j \in \mathcal{J}_0$, then $C^*_{1}/D^*_{1} > \gamma(1 - \zeta')$, i.e., full backing holds at $t=1$ if $\theta < \gamma(1 - \zeta')$.
\label{thm:corrupted_chain}
\end{theorem}

\subsection{Compromised Price Feeds} 
\label{sec:compromised_price_feeds}
Finally, we consider an attack in which a dishonest minority of price oracles tries to corrupt the price maintained by a vault. This could cause CDPs to be liquidated even if they are not underwater, thus violating safety of collateral. \rdc{We note that the DSS already has implemented safety
mechanisms for corrupted price oracles, e.g., its Oracle Security Module (OSM). The OSM\footnote{\url{https://docs.makerdao.com/smart-contract-modules/oracle-module/oracle-security-module-osm-detailed-documentation}} increases the available time to detect corrupted price feeds by adding an artificial delay before price information is accepted by the system. Such safety mechanisms are complementary to our analysis in the sense that they make individual price oracles harder to compromise, whereas our analysis shows that attacks against the system as a whole have negligible probability of success if only a minority of oracles is compromised. Furthermore, \name{}’s inclusion of global price oracles on minor chains provides another complementary safety mechanism. After all, if price oracles from established global sources are included, then it is harder for attackers to compromise a majority among the full set of oracles than if this set only includes local sources from the minor chain’s less developed DeFi ecosystem.}

Let $p$ be the base price for token $m \in \rda{M}$ at  time $t \geq 0$ (we omit $m$ and $t$ from the notation in this section as they are not relevant to the analysis). Let \rda{$O$ denote the set of all oracles} and $p_o$ the price returned by oracle $o \in O$, so that $p_o = p + \Delta_o$ \rda{where $\Delta_o$ is a normally distributed random variable as per \Cref{sec:threat_model}}. In the following, we assume some subset $O' \subset O$ is controlled by malicious parties who aim to convince \rda{a \name{} vault to accept} a price $p'$ that is different from the base price $p$. Let $p^*$ be the actual median price among all price feeds, including those that have been distorted. We can then prove the following theorem (see \Cref{sec:proofs}):

\begin{theorem}
If $|O'| < M/2$, then the probability that the difference between the base price $p$ and the median price $p^*$ exceeds a constant $c > 0$ vanishes exponentially as $c$ increases.
\label{thm:corrupted_price_feeds}
\end{theorem}

\subsection{Comparison to Baselines}
\label{sec:baseline_comparison_analysis}

In a fully siloed baseline, each coin chain has its own DSS-like stablecoin. 
If the chain is not compromised, then local and cross-chain transfers are secure (i.e., tokens cannot be stolen or locked permanently), whereas they are not if the chain is compromised. In this sense, the baseline is similar to \name{}.
The fully siloed system is unaffected by price crashes or platform failures on other blockchains -- however, it is less resilient to price crashes and platform failures on its own chain than a system with more diversified collateral: e.g., in \name, a stablecoin holder on a minor chain that experiences a price collapse would still be able to transfer stablecoins to another chain and exchange them for other tokens. Finally, because \name{} supports additional price oracles, it is more difficult for an adversary to compromise a majority of oracles than in the baseline.

The baseline that consists of the DSS with a wider set of wrapped tokens has the same security guarantees as \name{}, as it is similarly able to diversify its collateral. However, we will see in \Cref{sec:implementation} that this baseline has lower efficiency.

\section{Financial Analysis}
\label{sec:financial_analysis}
The main system requirement from \Cref{sec:requirements}, full backing, can be violated not just due to critical attacks and platform failures, but also due to natural price fluctuations in token prices. This vulnerability to price fluctuations is compounded by the strong correlation between the prices of different cryptocurrencies (see \Cref{fig:all_correlations} in the appendix). In this section, we use a financial risk model to estimate the probability that the system violates full backing over a 1-day period. 
As is common in financial analysis, we will refer to a basket of collateral tokens as a \textit{portfolio}, and to the relative increase in value of a token type over one time slot as a \textit{return}.
For our risk analysis, we compare a collateral portfolio that approximates the collateral composition of the existing DSS (without other stablecoins) to a portfolio made up of a wider range of blockchain tokens. 
In the following, we first present our model for token prices, and then present our risk analysis for the DSS-based portfolios.

\subsection{Token Price Model} 
\label{sec:token_price_model}

\rda{As mentioned in \Cref{sec:security_analysis}, we denote the base price of token $m \in M$ at time $t > 0$ by $p_{m\,t}$.} To model the evolution of token prices over time, we use \textit{Geometric Brownian Motion} (GBM). This is a common choice in the financial literature on asset prices \cite{ross2011elementary} and related work on stablecoins, e.g., by Salehi et al.\ \cite{salehi2021red}, particularly due to the GBM's role in the Black-Scholes formula for option prices.  A stochastic process $S_t$
that follows GBM is defined using a stochastic differential equation:
\begin{equation}
dS_t = \mu S_t dt + \sigma S_t d W_t,
\label{eq:gmb_definition}
\end{equation}
where $\mu$ determines the drift of $S_t$ (i.e., its expected change per unit of time), $\sigma$ represents the volatility of $S_t$ (i.e., the variance of the change), and
$W_t$ is the Wiener process that represents standard Brownian motion \cite{ross2011elementary}. 
In this model, each token price process $p_{m\,t}$, $t \geq 0$, can be represented using a process defined as in \Cref{eq:gmb_definition}. Essentially, this means that the logarithm of the relative difference of token prices, $\log\left(p_{m\,t+t'} / p_{m\,t}\right)$, is normally distributed for all $t' > 0$ with mean and variance dependent on $t'$.
Unlike \cite{salehi2021red}, we consider a setting with multiple assets; i.e., we have a stochastic process $S_{m\,t}$ for each token $m$ with its own mean $\mu_m$ and variance $\sigma^2_m$. In this setting, we must also consider
the \textit{covariance} between the token prices \cite{ross2011elementary}, as cryptocurrency prices are often highly correlated -- we will discuss this in more detail below. 

The standard GBM as decribed above is based on the normal distribution, i.e., $\log(S_{t+t'}/S_t)$ has a normal distribution. However, it is known from the academic literature that the probability distribution of stock price fluctuations is usually more heavy-tailed and that Student’s $t$-distribution is often a better fit \cite{kelly2014tail,li2020review}.  This is also noticeable in the data for price changes in our dataset.
As such, we will focus on the $t$-distribution in our analysis. Although the $t$-distribution does not allow for an analytical solution for the asset price change after each individual time slot,
we can use Monte Carlo simulation to analyze its behavior, as we discuss in more detail below.

\subsection{Risk Analysis: Approach}
\label{sec:financial_analysis_approach}

In this section, we investigate the probability that \name remains fully backed even in the presence of standard price fluctuations. In particular, we will estimate this probability for several choices of supported collateral tokens: 1) a single token type, 2) tokens that exist on the Ethereum blockchain, and 3) tokens that exist on a variety of independent blockchains. The second probability corresponds to the safety of the single-chain DSS, whereas the third option corresponds to the safety of \name. The purpose of this section is to demonstrate that the third option has considerably lower risk than the second, using multi-chain collateral baskets that are similar to those observed for the DSS in practice (see \Cref{tab:dss_parameters} in \Cref{sec:dss_implementation}). We discuss our approach in \Cref{sec:financial_analysis_approach} and present our empirical results in \Cref{sec:financial_analysis_evaluation}.


\textbf{Problem Statement.} Let $\mathcal{J}_{m\,t} = \{j \in \mathcal{J}_m:y_j = m\}$ be the set of CDPs at time $t$ whose token type is $m$, and $\alpha_{m\,t} = \sum_{j \in \mathcal{J}_{m\,t}} c_{j\,t}$ be the total number of collateral tokens of type $m$ at time $t$. In the following, we assume for simplicity that $\alpha_{m\,t} \equiv \alpha_m$, i.e., that the composition of the portfolio does not change during the observation period. This is a reasonable assumption as interactions with CDPs do not occur frequently.
In our model, failure of the stablecoin system occurs if, for any $t$ we have
\begin{equation}
\sum_{m \in M} \alpha_m p_{m\,t} \leq \theta D_t^*,
\label{eq:simulation_failure_criterion}
\end{equation}
where $\theta$ is the required over-collateralization and $D_t^*$ as defined in \Cref{sec:security_analysis}. We introduce a new variable $\gamma'$ that denotes the initial over-collateralization calculated as
$
\gamma' = \frac{\sum_{m \in M} \alpha_{m} p_{m\,t}}{D_{\rdc{0}}^*}.
$
Based on that, we can reformulate the failure condition as:
$
\frac{\sum_{m \in M} \alpha_{m} p_{m\,t}}{\sum_{m \in M} \alpha_{m} p_{m\,0}} \leq \frac{\gamma'}{\theta}.
$
As one would expect, failure does not depend on the absolute size of the portfolio, but on its relative composition.
Our first goal is then to estimate the probability of failure for a choice of $\alpha_1,\ldots,\alpha_{|M|}$, $\gamma'$, and $\theta$. 

In our simulation, we represent the portfolio as relative proportions of different assets. We define a vector $\vec{v}$ that denotes the value of
individual assets in the portfolio as
$
\textbf{v} = (v_1,v_2,\ldots,v_{|M|})^{T},
$
where $v_i$ is value of asset $i$. Similarly, we define a vector of random variables for log returns of individual assets
$
\textbf{X} = (X_1,X_2, \ldots,X_{|M|})^{\intercal},
$
where $X_m$ is the random variable for log returns of asset $m \in M$.
We then simulate the price of each asset as an individual GBM with Student's $t$-distribution for the return after each time slot, as discussed in \Cref{sec:token_price_model}. To accurately simulate price fluctuations, we must consider the correlation between the prices
-- to do so, we use Cholesky decomposition to compute a matrix $\mathbf{L}$ that satisfies
$
\textnormal{Cov}(\mathbf{X}) = \mathbf{L} \mathbf{L}^{T}
$. A matrix of random samples $\mathbf{R}$ can then be
generated as
$
\mathbf{R} = \mathbf{LT}
$,
where the matrix $\mathbf{T}$ contains a $t$-distributed random sample for each asset type and simulated time slot. 

\subsection{Risk Analysis: Evaluation}
\label{sec:financial_analysis_evaluation}

In this section, we analyze the risk of system failure in both single-chain and multi-chain settings using the approach described above. We first describe the data for our model, then the portfolio compositions, and finally the results of our analysis.

\textbf{Description of Data.}
We have collected data using the CoinMarketCap API over six different periods that span a total of 1264 hours of data, separated into 5-minute slots. The duration of these periods can be found in \Cref{tab:empirical_data_table} in \Cref{sec:price_data_table}. The final period includes the collapse of the FTX exchange on November 11, 2022. Our main motivation was to choose a diverse range of market conditions, including a period of turmoil, to ensure that the estimated system failure probabilities represent a worst-case scenario (i.e., in more typical market conditions, \name{} would be safer). In each time slot, we collected price information of 100 tokens, namely those that were the 100 largest in terms of market capitalization during the first time slot. A complete list of the included token symbols, and the observed correlation between them, can be found in \Cref{fig:all_correlations} in \Cref{sec:price_data_table}.

\textbf{DSS-Based Portfolios.}
For our analysis, we have considered single-asset portfolios consisting of Bitcoin and several of the highest-value smart contract platforms, namely: Ethereum, Cardano, Solana, Polkadot, Tron, Avalanche, Algorand,
and EOS -- their cryptocurrency symbols are: BTC, ETH, ADA, SOL, DOT, TRX, AVAX, ALGO, and EOS. The results are displayed in \Cref{tab:results_t} (single-asset). We use the following multi-asset portfolios: 
\begin{enumerate}[leftmargin=15pt]
\item D-All, which approximates the collateral portfolio of the DSS as of early 2023 (see \Cref{tab:dss_parameters} in \Cref{sec:dss_implementation}). (We replace GUSD, which was not in our dataset, with another fiat-backed stablecoin (USDT), and Maker's real-world assets with gold-backed PAXG.) 
\item D-ERC, which approximates the DSS portfolio after removing all fiat-backed assets.
\item D-Mix1 and D-Mix2, which use the same tokens as D-ERC but which increase the proportion of the minor ERC-20 tokens.
\item C-Mix1 and C-Mix2, which replace the ERC-20 tokens in D-Mix1 and D-Max2 with tokens from other popular blockchain platforms: ADA, DOT, TRX, and SOL.
\end{enumerate} 
 \rdc{We use these portfolios to show how the system performs with a mix of different types of collateral: other (fiat-backed) stablecoins, major cryptocurrencies (e.g., ETH/WBTC), and minor ERC-20 tokens. Tokens of the first type typically have much lower volatility, but introduce a dependence on trusted third parties. We will discuss this trade-off in more detail later in this section.}
The exact compositions of the portfolios can be found in \Cref{tab:portfolio_compositions} in \Cref{sec:portfolio_compositions}.

 In our experiments, 
 the expected over-collateralization to avert system failure is $\theta = 1.1$. For each table entry, we calculate
the probability of failure within 1 day, i.e., 288 5-minute time slots, using 100\,000 simulation runs. In each run, {we simulate the evolution of token prices by randomly drawing, in each consecutive time slot, the price change for each token type in the portfolio. 
We then compute the risk of failure as the fraction of runs that resulted in failure among the total number of runs.} 
The results are displayed in \Cref{tab:results_t}. \rdc{Since the 95\% confidence intervals widths are very small for the results in this table, we have omitted them for brevity.}

\begin{table}[b]
\centering
\setlength\tabcolsep{0.3em}
\caption{Chance of failure for single-asset portfolios and initial over-collateralization $\gamma'$.}
\label{tab:results_t}
\makebox[\textwidth][c]{
\scalebox{0.9}{
\begin{tabular}{c|ccccccccc|cccccc}
\multicolumn{1}{c}{} & \multicolumn{9}{c}{single-asset} & \multicolumn{6}{c}{multi-asset} \\
$\gamma'$ & \rot{BTC} &  \rot{ETH} &   \rot{ADA} &   \rot{SOL} &   \rot{DOT} &   \rot{TRX} &  \rot{AVAX} &  \rot{ALGO} &   \rott{EOS} & \rot{D-All} &  \rot{D-ERC} &  \rot{D-Mix1} &  \rot{D-Mix2} &  \rot{C-Mix1} &  \rot{C-Mix2} \\ 
\midrule
1.2 & 0.021 & 0.047 & 0.078 & 0.228 & 0.063 & 0.018 & 0.126 & 0.166 & 0.097 & 0.001 & 0.051 & 0.040 & 0.050 & 0.030 & 0.029 \\
1.3 & 0.003 & 0.008 & 0.014 & 0.093 & 0.012 & 0.003 & 0.029 & 0.051 & 0.016 & 0.000 & 0.007 & 0.005 & 0.006 & 0.003 & 0.003 \\
1.4 & 0.001 & 0.003 & 0.005 & 0.059 & 0.004 & 0.001 & 0.009 & 0.019 & 0.004 & 0.000 & 0.002 & 0.002 & 0.002 & 0.001 & 0.001 \\
1.5 & 0.001 & 0.001 & 0.002 & 0.041 & 0.002 & 0.001 & 0.004 & 0.007 & 0.002 & 0.000 & 0.001 & 0.001 & 0.001 & 0.001 & 0.001 \\
\bottomrule
\end{tabular}
}
}
\setlength\tabcolsep{0.5em}
\end{table}

\begin{table}[b]
\centering
\setlength\tabcolsep{0.3em}
\caption{\rdc{Chance of failure for single-asset portfolios and initial over-collateralization $\gamma'$. In the table, `0.0 $\pm$ 0.0' indicates that system failure events were observed during the 1,000,000 simulation runs, but that the estimated probability has been rounded down to $0$, whereas `---' indicates that no failure events were observed.}}
\label{tab:results_stablecoins_t}
\makebox[\textwidth][c]{
\scalebox{0.9}{
\rdc{
\begin{tabular}{l|ccccc|ccccc}
\multicolumn{1}{c}{} & \multicolumn{5}{c}{single-asset} & \multicolumn{5}{c}{C-Mix1, stablecoin contribution $\delta$} \\
 $\gamma'$ & USDT & USDC & BUSD & TUSD & DAI & $\delta=0.5$ & $\delta=0.25$ & $\delta=0.1$ & $\delta=0.05$ & $\delta=0$ \\
\midrule
vol. & 
    4.85e-5 & 
    9.70e-5 & 
    4.06e-4 & 
    3.48e-4 & 
    2.84e-4 & 
    7.86e-4 & 
    1.17e-3 & 
    1.39e-3 & 
    1.47e-3 & 
    1.55e-3 \\ \midrule
1.2 & 
    \confintv{0.00001}{2.2e-6} & 
    \confintv{0.00009}{9.2e-6} & 
    \confintv{0.00177}{4.2e-5} & 
    \confintv{0.00103}{3.2e-5} & 
    \confintv{0.00078}{2.8e-5} & 
    \confintv{0.00314}{5.6e-5} & 
    \confintv{0.01198}{1.1e-4} & 
    \confintv{0.02153}{1.5e-4} & 
    \confintv{0.02548}{1.6e-4} & 
    \confintv{0.02976}{1.7e-4} \\
1.3 & 
    \confintv{0.0}{0.0} & 
    \confintv{0.00002}{4.7e-6} & 
    \confintv{0.00041}{2.0e-5} & 
    \confintv{0.00027}{1.6e-5} & 
    \confintv{0.00023}{1.5e-5} & 
    \confintv{0.00047}{2.2e-5} & 
    \confintv{0.00137}{3.7e-5} & 
    \confintv{0.00243}{4.9e-5} & 
    \confintv{0.00304}{5.5e-5} & 
    \confintv{0.00345}{5.9e-5} \\
1.4 & 
    --- & 
    \confintv{0.00001}{3.0e-6} & 
    \confintv{0.00021}{1.4e-5} & 
    \confintv{0.00012}{1.6e-5} & 
    \confintv{0.00009}{9.3e-6} & 
    \confintv{0.00018}{1.3e-5} & 
    \confintv{0.00052}{2.3e-5} & 
    \confintv{0.00098}{3.1e-5} & 
    \confintv{0.00105}{3.2e-5} & 
    \confintv{0.00119}{3.5e-5} \\
1.5 & 
    \confintv{0.0}{0.0} & 
    \confintv{0.00001}{2.8e-6} & 
    \confintv{0.00012}{1.1e-5} & 
    \confintv{0.00001}{8.3e-6} & 
    \confintv{0.00001}{7.7e-6} & 
    \confintv{0.00009}{9.6e-6} & 
    \confintv{0.00029}{1.7e-5} & 
    \confintv{0.00046}{2.1e-5} & 
    \confintv{0.00033}{1.8e-5} & 
    \confintv{0.00065}{2.5e-5} \\
\bottomrule
\end{tabular}
}
}
}
\setlength\tabcolsep{0.5em}
\end{table}

From \Cref{tab:results_t}, we observe a considerable difference between the failure probabilities of the single-token portfolios. The most high-risk portfolio was Solana (SOL), which was severely affected by the collapse of FTX, since it was one of FTX's most prominent backers \cite{business_insider}. Even for an initial collateralization  $\gamma'=1.5$, a pure SOL portfolio would have a failure probability of more than $4\%$ after a day, given the data from our observation period. In contrast, BTC and TRX were the most stable cryptocurrencies over this period.

\begin{figure}
    \makebox[\linewidth][c]{
    	\subfloat[Different blockchains]{%
    		\includegraphics[width=0.525\linewidth]{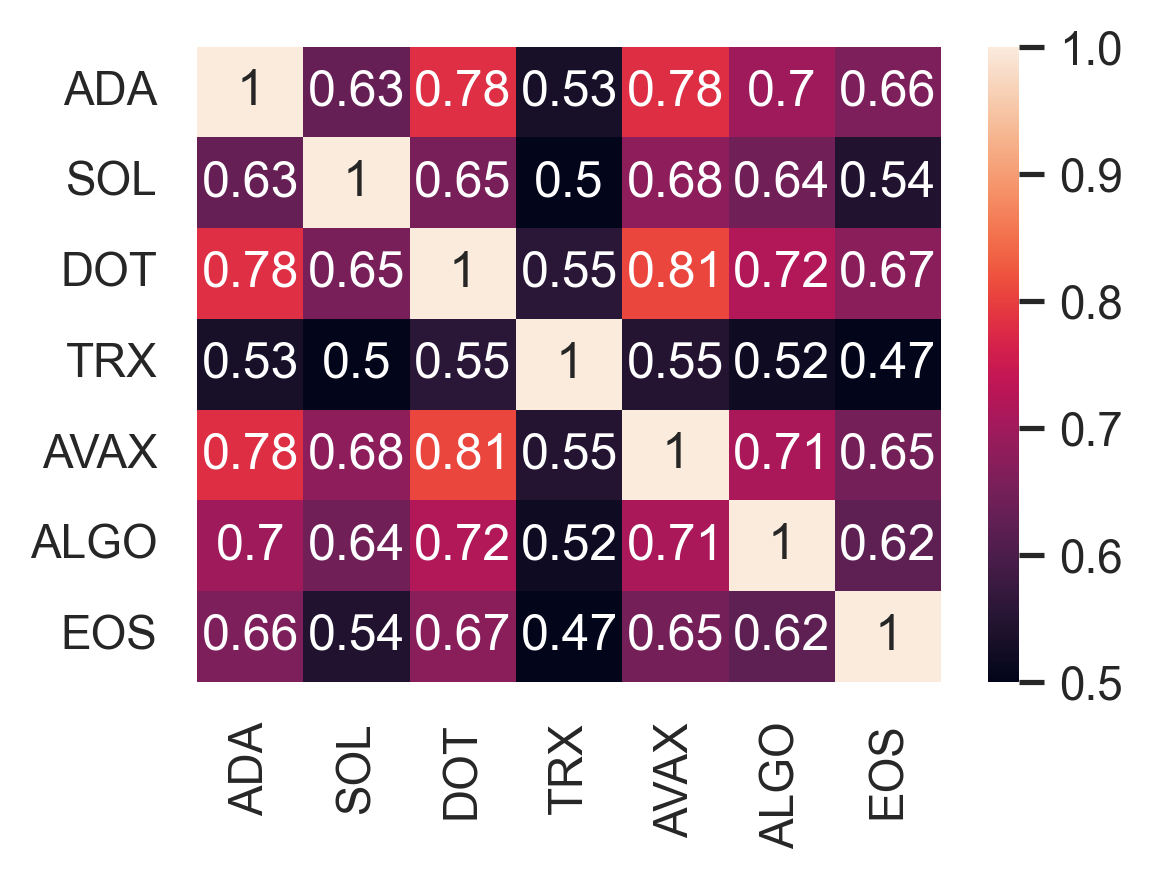}
    	\label{fig:correlations_blockchains} }
    	\subfloat[ERC20 tokens]{%
    		\includegraphics[width=0.525\linewidth]{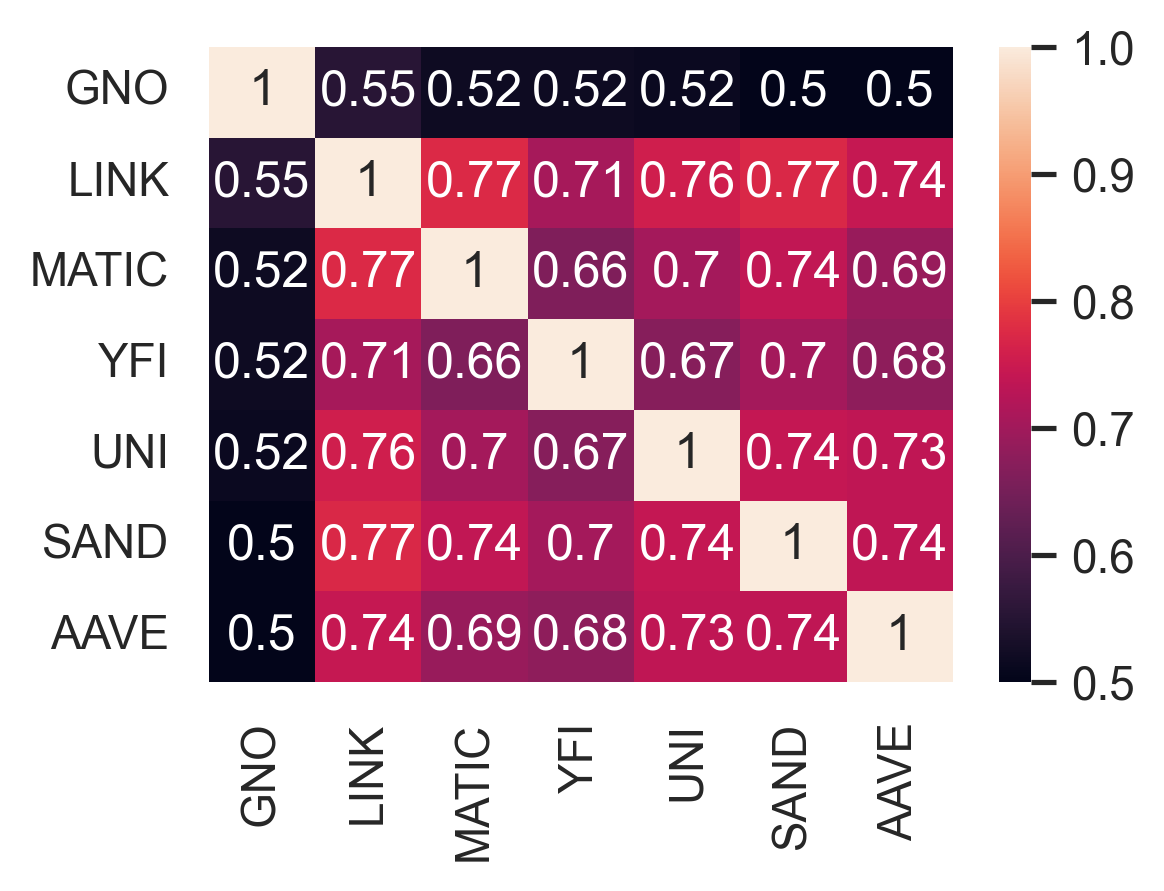}
    	\label{fig:correlations_erc20} }
    }
	\caption{Correlation between the different token types. 
 }
 \label{fig:correlation_heatmaps}
\end{figure}

Furthermore, we observe from \Cref{tab:results_t} that the DSS's approximate portfolio has a very low probability of failure due to $80\%$ of its collateral being stablecoins and other pegged assets. 
In fact, for $\gamma'=1.2$, its observed failure probability (0.001 for D-All) was roughly 40-50 times lower than for the pure crypto-backed portfolios (0.051, 0.040, and 0.050 respectively for D-Erc, D-Mix1, and D-Mix2).
For the pure crypto-backed portfolios, we find that D-ERC has slightly lower failure probabilities than a pure ETH portfolio, which can be explained by the higher share of low-volatility WBTC in the portfolio. Similarly, C-Mix1 and C-Mix2 portfolios have lower risk than the D-ERC due to the inclusion of more WBTC. We do not observe much benefit from the inclusion of more different asset types due to the high correlation between the included tokens. However, when comparing the C-Mix1 and C-Mix2 portfolios we do note that C-Mix2 has a lower failure probability despite the inclusion of less low-volatility WBTC. The reason is the inclusion of TRX, which is not only similarly low-risk as WBTC, but which also has a relatively low correlation with the other asset types -- this can be seen  in \Cref{fig:correlations_blockchains}. We also note that the relative differences become starker when $\gamma'$ increases -- e.g., for $\gamma'=1.2$, the risk of D-Mix1 is roughly a third higher than C-Mix2 (4.0\% vs.\ \mbox{2.9\%}), whereas for $\gamma'=1.3$ the risk is two thirds higher (0.5\% vs. 0.3\%).

\rdc{\textbf{Comparison to Existing Stablecoins.}} \rdc{
From \Cref{tab:results_t}, we observe that D-All has a considerably lower failure risk than the other portfolios.
D-All represents the collateral portfolio of the DSS and  is the only multi-asset portfolio from \Cref{tab:results_t} that includes other stable assets (around 75\% as per \Cref{tab:dss_parameters}).
All stablecoins except Dai in our dataset are fiat-backed stablecoins, which exhibit much smaller price volatility than regular tokens, but rely on trusted third parties which expose the system to external black swan events. We conduct experiments to further understand the price volatility. In particular, we estimate the failure risks for single-stablecoin portfolios and for a range of DSS-based portfolios with varying degrees of included stablecoins. For the latter, we consider modifications of the C-Mix1 portfolio, which represents a portfolio similar to Dai's collateral portfolio but which includes a minor fraction (20\%) of tokens from other blockchains. We modify this portfolio by assigning a fraction $\delta$ of fiat-backed stablecoins to it and multiplying the contribution of the other token types by $1-\delta$, for varying values of $\delta$. The low volatility of other stablecoins means that the failure probabilities of single-stablecoin portfolios are very low. Therefore, we conduct 1,000,000 runs for each table entry and display the 95\% confidence interval half-width below each estimate in \Cref{tab:results_stablecoins_t}. We also display the variance, which is strongly related to the failure risks, of each portfolio in the table (under vol.).
}

\rdc{
From \Cref{tab:results_stablecoins_t}, we observe that the failure risks of all single-stablecoin portfolios are much lower than those of the single-asset portfolios in \Cref{tab:results_t}, although there are still considerable differences between individual stablecoins: e.g., BUSD's volaility is 10 times higher than USDT's, and, for $\gamma'$, BUSD's estimated 1-day probability of failure is 177 times higher than that of USDT. For the modified C-Mix1 portfolios, we have used USDC as a representative stablecoin because of its high trade volume (there have been over $20$bn USDC in circulation since July 2021) and because its volatility is neither exceptionally high nor low among the included stablecoins. It can be seen that, as expected, the inclusion of other stablecoins leads to low failure risks, although we observe a small decrease for higher values of $\gamma'$: in particular,  the relative difference between $\delta=0$ and $\delta=0.25$ is roughly 2.5$\times$ (0.02976 vs.\ 0.01198), but this decreases to 2.24$\times$ for $\gamma'=1.5$. For $\delta=0$ and $\delta=0.5$, the relative difference decreases from 9.48$\times$ for $\gamma'=1.1$ to 7.22$\times$ for $\gamma'=1.5$. Overall, we observe that the inclusion of stablecoin fraction of $25\%$ can reduce the system's failure risk by more than $60\%$. }

\rdc{We note that a stablecoin system that is overcollateralized through a \textit{mixture} of fiat-backed stablecoins is more resilient to black-swan events than the individual stablecoins. The reason is that the collapse of any \textit{single} trusted third party would have a limited impact. However, black swan events may affect multiple fiat-backed stablecoins simultaneously, e.g., due to new legislation affecting stablecoin collateral. Thus, it would still be advisable to ensure that the \textit{total} contribution of fiat-backed stablecoins is limited. 
We also note, that, in theory, a principal crypto-backed stablecoin could include other crypto-backed stablecoins as collateral. In this setting, the portfolio of the supporting stablecoins should be included in the portfolio of the principal stablecoin. 
Finally, protocol designers with substantial fiat assets (e.g., Dai) also have the option to steer collateral portfolios towards those with minimal risk \cite{hajek2024collateral}, which would reduce the risks associated with the crypto collateral even further.} 

\section{Implementation}
\label{sec:implementation}
In this section, we present a prototype implementation of \name and compare its performance \rda{in terms of both cost and processing time} to a baseline that combines the existing DSS with wrapped tokens.
 We present three different subvariants of the relay chain design with a quorum of  nodes (i.e., the offchain variant of \Cref{sec:relay_chain_designs}).
Each subvariant has a different tradeoff in terms of scalability and generality, and consists of two smart contracts written in Ethereum's Solidity language: \texttt{Relay} and \texttt{CCTransfer}, with only \texttt{Relay} different between the three subvariants. The \texttt{Relay} and \texttt{CCTransfer} contracts interact with an ERC-20 compliant contract for token transfers, \texttt{Coin}, which re-uses code from the \texttt{Dai} contract in the DSS \cite{daisol}. We discuss the smart contracts in more detail in \Cref{sec:smart_contracts}.
The other components depicted in \Cref{fig:crocodai} (i.e., the vault, auction, and oracle contracts) are also re-used from the DSS. 
We perform two sets of experiments: microbenchmarks for the relay subvariants with varying quorum sizes, and end-to-end experiments for \name{}'s main operations: creating new tokens, and cross-chain transfers.

\subsection{Smart Contracts}
\label{sec:smart_contracts}

\hspace{0.35cm}\textbf{CCTransfer. } This contract stores cross-chain transfer requests as a hashmap of \texttt{Request} instances, which are implemented using a Solidity \texttt{struct}. Each \texttt{CCTransfer} contract is linked to a specific \texttt{Coin} and \texttt{Relay} contract pair, and must be made a \texttt{Coin} contract ward to allow it to mint or burn stablecoins upon the commitment of cross-chain transfers. 

\textbf{Relay. } In the first design (named $n$/$n$), a hashmap of relay node addresses in Solidity's \texttt{address} format is stored in the contract. For each cross-chain transfer, at least $2f+1$ relay nodes must call the \texttt{Relay} contract's \texttt{approve} function using their stored address. This design exploits the EVM's built-in signature verification, and generalizes trivially to other blockchain platforms. However, the downside is that each relay node must have an account and spend gas on each supported blockchain for every cross-chain transfer.
In the second design ($n$/$1$), a hashmap of relay node public keys is stored in the contract, and a batch of $2f+1$ 32-byte signatures is passed in the form of a byte array as a function parameter for each transaction. These signatures are then validated one by one by repeatedly calling the built-in \texttt{ecrecover} function. Although this is not as generic as the previous approach (e.g., each chain with its different cryptographic primitives would require its own off-chain signature generation script)
it is more efficient and does not require each individual relay node to maintain an account.
In the third design ($1$/$1$), we use a $t$-out-of-$n$ threshold signature scheme (where in our case $t=2f+1$) to produce a single 32-byte signature for the relay node quorum, e.g., using the procedure of \cite{gennaro2018fast}. The resulting signature is a standard ECDSA signature that can be verified using \texttt{ecrecover}. For our implementation we have used the multi-party ECDSA library by ZenGo X \cite{zengox}.

\begin{table}[b]
\centering
\caption{Performance comparison of the three implementation subvariants. }
\label{tab:implementation_choices}
\scalebox{0.9}{
\begin{tabular}{c|ccc|ccc}
 & \multicolumn{3}{c|}{gas cost} & \multicolumn{3}{c}{$T_{\textnormal{proc}}$ (s)} \\ 
design & $n=4$ & $n=10$ & $n=16$ & $n=4$ & $n=10$ & $n=16$ \\ \toprule
(1) $n$/$n$ & 171791 & 511473 & 977350 & {---} & {---} & {---} \\
(2) $n$/$1$ & 214373 & 430690 & 647118 & 0.025 & 0.071 & 0.091 \\
(3) $1$/$1$ & 53112 & 53117 & 53110 & 1.101 & 7.780 & 34.513 \\ \bottomrule
\end{tabular}
}
\end{table}

\subsection{Relay Microbenchmarks}

\Cref{tab:implementation_choices} compares the different designs in terms of the gas cost and the processing time at the relay nodes, $T_{\textnormal{proc}}$. Both are essential to system performance: if the on-chain cost is too high then each transfer will incur large fees, whereas if the computation time is too high then cross-chain transactions incur a large delay. For the experiments, we used a Lenovo IdeaPad with an Intel(R) Core(TM) i5-1035G1 CPU (1.00GHz) and 8GB RAM -- however, the gas costs are machine-independent. The scripts that performed off-chain computations use Python's Web3 package for cryptographic operations, and the underlying blockchain was simulated using Truffle's Ganache tool \cite{ganache}. \rdc{To execute our experiments, we ran the Python scripts in different terminals that each interacted with an active Ganache instance on a single machine.} From \Cref{tab:implementation_choices}, we observe that the first subvariant is the worst in terms of performance, which is the cost of its generality and ease of implementation. The second subvariant is more efficient in terms of gas costs for large quora, and only has little processing time at the relay nodes. The third subvariant is the most efficient in terms of gas costs, but at the cost of considerable processing time at the relay nodes, especially for large quora. For the first  subvariant (i.e., $n/n$), we were unable to distinguish $T_{\textnormal{proc}}$ from the transaction processing time in Ganache as both happen in a single step in Python's Web3 package -- however, these are assumed to be similar to signing in the $n/1$ design, and are indicated by `---' in the table.


\subsection{End-to-End Experiments}
\label{sec:end_to_end_experiments}

As discussed in \Cref{sec:baseline_comparison_analysis}, a baseline that uses the existing DSS with a greater variety of wrapped tokens has similar resilience to black swan events and price fluctuations as \name{}. 
However, we provide experimental results in this section that demonstrate that the DSS baseline has worse end-to-end performance than \name{}. We focus on the two main stablecoin actions: 
\begin{inparaenum}
\item creating new stablecoins on a coin chain and 
\item transferring stablecoins from one coin chain (Chain A) to another (Chain B). 
\end{inparaenum} 
To ensure that gas costs are comparable, we implement Chain A, Chain B, and the relay chain as private Ethereum chains using \texttt{geth}.
For collateral, we implement a custom ERC20 token contract on the underlying chain (\texttt{GemCoin}) -- for the DSS baseline, we also require a bridge for the custom token on Chain A (\texttt{GemBridge}), and another (\texttt{DaiBridge}) for Dai on the relay chain. For the collateral vault, we use instances of \texttt{Vat}, \texttt{DaiJoin}, and \texttt{GemJoin} where the latter is linked to the \texttt{GemCoin} contract -- in \name{} the vault contracts exist on Chain A, whereas in the DSS baseline they exist on the relay chain. Finally, the baseline requires a  \texttt{WrGemCoin} contract on the relay chain to implement the wrapped collateral token as an ERC20 token, and a \texttt{WrDai} contract for wrapped Dai tokens on each coin chain.



\textbf{Creating Stablecoins. }
For \name{}, the procedure for creating new stablecoins is as follows:
the user first transfers the collateral to the vault and creates a CDP using \texttt{GemCoin} and \texttt{GemJoin}.
Next, she  uses the CDP to create stablecoins and transfers them to her account using \texttt{Vat} and \texttt{DaiJoin}.

For the DSS baseline, the procedure is as follows.
The user first sends GemCoin tokens to the bridge using \texttt{GemCoin} and \texttt{BaseBridge}.
Next, a relay detects the transfer and creates an equivalent amount of wrapped GemCoin tokens for the user using \texttt{WrGemCoin}. 
The user then transfers the collateral to the vault, creates a CDP, creates stablecoins, and  transfers them to her account using \texttt{WrGemCoin}, \texttt{GemJoin}, \texttt{Vat} and \texttt{DaiJoin}.
Next, the user sends the Dai tokens to the bridge using \texttt{Dai} and \texttt{DaiBridge}.
Finally, a relay detects the transfer and mints an equivalent amount of wrapped Dai tokens for the user using \texttt{WrDai}. 

\textbf{Cross-Chain Transfers. }
For \name{}, the procedure for transferring stablecoins from Chain A to B is as follows.
The user first transfers the tokens to \texttt{CCTransfer} and creates a transfer request.
Next, a relay detects the request, sends a signed approval to the relay chain and \texttt{Relay} on Chain A, which calls \texttt{Coin} to burn the stablecoins.
Finally, the relay sends the signed message to \texttt{Relay} on Chain B, which calls \texttt{Coin} to mint the stablecoins.

For the DSS baseline, the procedure is as follows:
the user first burns her wrapped Dai tokens using \texttt{WrDai}.
Upon detection by a  relay, the burned Dai is refunded to the user on the relay chain.
The user then transfers the Dai to the \texttt{DaiBridge} for Chain B.
Finally, a relay detects the transfer and mints an equivalent amount of wrapped Dai tokens for the user using \texttt{WrDai} on Chain B.

\newcommand{\dbc}[1]{#1}
\newcommand{\dbcb}[1]{\textbf{#1}}

\setlength{\tabcolsep}{4pt}
\begin{table}[t]
\centering
\caption{Computation costs (units of gas) of the transactions in \name and the DSS baseline, grouped by block. \textbf{Bold}: relay chain computations.}
\label{tab:end_to_end_experiments}
\vspace{0.2cm}
\begin{tabular}{lc|cccc}
 & & \multicolumn{2}{c}{\name} & \multicolumn{2}{c}{DSS baseline} \\ 
 & & creation & transfer & creation & transfer \\ \toprule
Block 1 & (Chain A) & 
    \dbc{466537} & 
    \dbc{229325} & 
    \dbc{244036} &  
    \dbc{27762} \\
Block 2 & \textbf{(Relay)} & 
     & 
    \dbcb{133280} & 
    \dbcb{770724} & 
    \dbcb{294016} \\
Block 3 & (Chain B) & 
    & 
    \dbc{185503} & 
    & 
    \dbc{71498} \\
Block 3' & (Chain A) & 
    & 
    & 
    \dbc{71498}  & 
    \\ \midrule
Total & (coin chains) & 466537 & 414828 & 315534 & 99260 \\
\textbf{Total} & \textbf{(relay chain)} & \textbf{0} & \textbf{133280} & \textbf{770724} & \textbf{294016} \\
\bottomrule
\end{tabular}
\end{table}

\textbf{Results. }
\Cref{tab:end_to_end_experiments} displays the cost of the processes discussed above, with transactions that can occur within the same block grouped together.
For the DSS baseline, we assume that a single trusted custodian acts as a relay -- to keep the comparison with \name{} as fair as possible, we hence use the first prototype subvariant ($n/n$) but with a single signer. 
For our implementation, we assume that the recipient data for outgoing cross-chain transfers (i.e., the ID of the receiving chain and the recipient's account data) is not written to the blockchain; even in the fully on-chain design (see \Cref{sec:relay_chain_designs}), this data can be exchanged off-chain and incorporated in the commit message or zero-knowledge proof.

We observe that the DSS baseline has considerably higher gas costs (namely 1\,086\,258 vs.\ 466\,537 gas, or $\approx$$130\%$ more) for creating stablecoins, due to the expensive two-way bridge -- furthermore, it requires three blocks on two chains, whereas \name{} can perform all operations within a single Chain A block. This is a major difference because, to ensure security against fork attacks, blocks must be finalized before their information is processed on another chain -- e.g., using Ethereum's Casper FFG \cite{buterin2020incentives}, which currently takes between one and two 32-block epochs in Ethereum. As such, local stablecoin creation is considerably cheaper and faster in \name{} than in the baseline system, which allows \name{} to exploit the availability of local collateral without depending on it.
For cross-chain transfers, the total gas costs are slightly lower for the DSS baseline than for \name{} (393\,276 vs. 548\,108), but the DSS baseline consumes nearly twice as much gas on the relay chain (\textbf{bold}), which is in the DSS is Ethereum. This means that as long as gas is roughly 1.96$\times$ cheaper on the coin chains than on Ethereum, then \name has cheaper cross-chain transactions -- this is a reasonable assumption in practice, as lower transaction fees are a key feature of smaller chains. For example, transactions on Arbitrum and Optimism were typically at least 8$\times$ cheaper than on Ethereum in September 2023 \cite{l2fees}.











\section{Conclusions and Future Work}
\label{sec:conclusions}
We have presented \name, a stablecoin for cross-chain commerce. \name facilitates cross-chain transfers by design, and is more resilient than a siloed baseline to black swan events and price fluctuations. 
We have presented a prototype implementation that relies on a quorum of relay nodes for cross-chain transfers and governance decisions, and which has better performance than a DSS-based baseline in terms of (expensive) transactions on the relay chain.
For future work, an interesting direction is to implement a version of \name that performs all relay operations on-chain, as per the second variant of \Cref{sec:relay_chain_designs} -- e.g., using the approach of zkBridge, or succinct proofs of state as in Algorand or Mina.
Another interesting direction is to further investigate alternative collateral liquidation mechanisms -- for example,
 Salehi et al.\ \cite{salehi2021red}
 propose to mitigate liquidations from the DSS or remove them completely. 
A final interesting research direction is system recovery after a 51\% attack on a coin chain.

\section{Acknowledgements}
This work was supported by Ministry of Education (MOE) Singapore's Tier 2 Grant Award No. MOE-T2EP20120-0003.

\bibliographystyle{plain}
\bibliography{ref2}

\appendix


\section{Proofs of Security Theorems}
\label{sec:proofs}
Before we prove the main theorems, we need the following technical lemmas. We need the first to prove that if all CDPs are overcollateralized, then so is system as a whole.

\begin{lemma}
If $a_i/b_i > c$ for $i=1,\ldots,n$, then $\frac{\sum_{i=1}^n a_i}{\sum_{i=1}^n b_i} > c$.
\label{lm:convex_combination}
\end{lemma}

\begin{proof}
Note that $\frac{\sum_{i=1}^n a_i}{\sum_{i=1}^n b_i} = \sum_{i=1}^n \frac{a_i}{b_i} \frac{b_i}{\sum_{j=1}^n b_j}$. Since $\sum_{i=1}^n \frac{b_i}{\sum_{j=1}^n b_j} = 1$, $\frac{\sum_{i=1}^n a_i}{\sum_{i=1}^n b_i}$ is a convex combination of $a_i/b_i, \ldots, a_n/b_n$. As such, it cannot be smaller than $\min_{i=1,\ldots,n} a_i/b_i$, which is greater than $c$ by the statement of the lemma.
\end{proof}

The second lemma is a commonly-known result, but we include the proof here for completeness.

\begin{lemma}
If $X$ has a normal distribution with mean $0$ and variance $1$, then 
\begin{equation*}
R(x) = \mathbb{P}(X > x) \leq \frac{e^{-x^2/2}}{x\sqrt{2\pi}}.
\end{equation*}
\label{lm:normal_tail}
\end{lemma}

\begin{proof}
Since for $y \geq x \geq 0$ we have that $1 \leq \frac{y}{x}$, we have that
\begin{equation*}
\begin{split}
\mathbb{P}(X > x) & = \frac{1}{\sqrt{2\pi}} \int_{x}^{\infty} e^{-y^2/2} dy \\ 
& \leq \frac{1}{\sqrt{2\pi}} \int_{x}^{\infty} \frac{y}{x} e^{-y^2/2} dy = \frac{e^{-x^2/2}}{x\sqrt{2\pi}}.
\end{split}
\end{equation*}
\end{proof}
With these lemmas, we are ready to prove the three theorems from \Cref{sec:security_analysis}.
\begin{proof}[Proof of Theorem \ref{thm:sudden_price_drop}]

 If $p_{y_j\,t} c_{j\,t}/y_{j\,t} > \gamma$ holds for all $j \in \mathcal{J}$, then by \Cref{lm:convex_combination} it must hold that
\begin{equation}
\frac{C^{+}_0}{D^{+}_0} > \gamma.
\label{eq:unaffected_overcollateralized}
\end{equation}
Furthermore, we know from the definition of the debt ceiling that $D^{-}_0 < \zeta' (D^{-}_0 + D^{+}_0)$, where $\zeta'$ (as defined in \Cref{sec:security_analysis}) is the total debt ceiling of the affected tokens. The above can be rewritten to 
\begin{equation}
D^{-}_0 < \zeta'/(1-\zeta') D^{+}_0.
\label{eq:debt_ceiling_rewritten}
\end{equation}
We therefore find that
\begin{equation*}
\frac{C^*_{1}}{D^*_{1}} = \frac{C^{+}_{0}}{D^{-}_{0} + D^{+}_{0}} > \frac{C^{+}_{0}}{\left(1+\frac{\zeta'}{1-\zeta'}\right) D^{+}_{0}} > \gamma(1 - \zeta')
\end{equation*}
where the first inequality holds because of \eqref{eq:debt_ceiling_rewritten} and the second because of \eqref{eq:unaffected_overcollateralized}.
\end{proof}

\begin{proof}[Proof of Theorem \ref{thm:corrupted_chain}]
Because of the debt ceilings on the tokens, the total debt incurred for each token type $m$ must still satisfy $\sum_{\substack{j \in \mathcal{J}_t \\ y_j = m}} D_{j\,t} < \zeta_m \sum_{j \in \mathcal{J}_t} D_{j\,t}.$ -- as such, even after the attack, we must have 
$
D^{-}_{1} = D^{-}_{0} + h' < \zeta' (D^{+}_0 + D^{-}_0 + h')
$
or
\begin{equation}
D^{-}_0 + h' < \zeta'/(1- \zeta') D^{+}_0
\label{eq:debt_ceiling_rewritten_2}
\end{equation}
where $D^{-}_0$, $D^{+}_0$, and $\zeta^{-}$ are as defined in the previous section. We have that
\begin{equation*}
\frac{C^*_{1}}{D^*_{1}} = \frac{C^{-}_{0} + C^{+}_0}{D^{-}_{0} + D^{+}_{0} + h} > \frac{C^{-}_{0} + C^{+}_{0}}{\left(1 + \frac{\zeta'}{1-\zeta'}\right) D^{+}_0} > \gamma(1 - \zeta')
\end{equation*}
where the first inequality holds because of \eqref{eq:debt_ceiling_rewritten_2} and the second because $C^{-}_{0} > 0$ and because of \eqref{eq:unaffected_overcollateralized}. As such, we again find that the debt ceilings ensure that \name remains fully backed if it is sufficiently overcollateralized.
\end{proof}

\begin{proof}[Proof of Theorem \ref{thm:corrupted_price_feeds}]
In the following, we first assume that $O$ is odd and that some subset $O' \subset O$ is controlled by malicious parties who aim to convince a vault in \name that the price is at most than some $p^* > p$. In this case, the median equals the $\frac{1}{2} (M+1) + |O'|$ th order statistic from a sample of $M-|O|$ normally distributed random variables with mean $p$ and standard deviation $\sigma_o^2$. This order statistic is greater than the maximum of the $|M|-|O'|$ random variables, which we denote by $Y$. Let $\sigma_{*} = \min_{o \in O} \sigma_o $. $Y$'s cumulative distribution function equals 
\begin{equation*}
\begin{split}
\mathbb{P}(Y \geq p^*) & = \mathbb{P}\left(\bigwedge_{o \in O \backslash M'} \Delta_o > p^* - p\right) = \prod_{o \notin O'} R\left(\frac{p^*-p}{\sigma_o}\right)  \\ & \leq R^{|M|-|O|}\left(\frac{p^* - p}{\sigma_{*}}\right) \leq \frac{e^{-(|M|-|O'|)(p^* - p)^2/2}}{(p^* - p)\sigma\sqrt{2\pi}},
\end{split}
\end{equation*}
where the first inequality holds because $R$ is a monotonically decreasing function, and the second inequality follows from \Cref{lm:normal_tail}. As such, the probability $\mathbb{P}(Y \geq p^*)$ vanishes exponentially as $p^* - p$ increases. The proof for $p^* < p$ is very similar and uses the symmetry of the normal distribution.
\end{proof}




\section{DSS Parametrization}
\label{sec:dss_implementation}

The parameters of the various different vault types in the DSS can be found \textcolor{blue!50!gray}{\href{daistats.com}{here}}. A sample of parameter values as of 25 April 2023 (Ethereum block 17121982) can be found in \Cref{tab:dss_parameters}. For each of the supported token type $j$, we indicate the token name, total value of the collateral of CDPs of this type $\sum_{j \in \mathcal{J}_t} p_{j\,t} c_{j\,t}$, total debt of the collateral of CDPs of this type $\sum_{j \in \mathcal{J}_t} D_{j\,t}$, collateralization ratio, and debt ceiling $\zeta_j$. Although the exact values vary throughout the day, this table is meant to illustrate the typical composition of the collateral type in the DSS.

The following token types were supported in the DSS, but were later deprecated: AAVE, BAL, BAT, COMP, KNC, LRC, MANA, RENBTC, TUSD, UNI, USDT, and ZRX. The following liquidity pool tokens were deprecated: AAVE/ETH, LINK/ETH, UNI/ETH, wBTC/Dai, and wBTC/ETH.

\begin{table*}[ht]
    \centering    
    \caption{Information about the various collateral types used by the DSS (as of Ethereum block 17121982).}
    \label{tab:dss_parameters}
    \scalebox{0.65}{
    \begin{tabular}{clc|rrcrccc}
    & Name & Token/Asset Platform  & debt & coll.\ value & overcoll.\ & debt ceil.\ & stab.\ fee & liq.\ rate & contrib.\ \\ \toprule
\parbox[t]{2mm}{\multirow{13}{*}{\rotatebox[origin=c]{90}{Standard ERC-20 Tokens}}} & ETH-A & Ether & 233580929 & 421276897 & 180\% & 15000000000 & 1.50\% & 145\% & 4.69\% \\ 
 & ETH-B & ---  & 45480335 & 67797760 & 149\% & 250000000 & 3.00\% & 130\% & 0.91\% \\ 
 & ETH-C & --- & 306467435 & 407905059 & 133\% & 2000000000 & 0.75\% & 170\% & 6.16\% \\
 & WSTETH-A & Wrapped Ether & 105440104 & 123821333 & 117\% & 500000000 & 1.50\% & 160\% & 2.12\% \\
 & WSTETH-B & --- & 198020222 & 212522461 & 107\% & 500000000 & 0.75\% & 185\% & 3.98\% \\
 & RETH-A & --- & 3503590 & 6503590 & 186\% & 20000000 & 0.50\% & 170\% & 0.07\% \\
 & WBTC-A & Wrapped Bitcoin & 35684532 & 55869143 & 157\% & 500000000 & 1.75\% & 145\% & 0.72\% \\
 & WBTC-B & --- & 22535846 & 32535847 & 144\% & 250000000 & 3.25\% & 130\% & 0.45\% \\
 & WBTC-C & --- & 40106301 & 63565858 & 158\% & 500000000 & 1.00\% & 175\% & 0.81\% \\
 & LINK-A & ChainLink & 711324 & 3708567 & 521\% & 5000000 & 2.50\% & 165\% & 0.01\% \\
 & YFI-A & Yearn & 148553 & 1941526 & 1307\% & 4000000 & 1.50\% & 165\% & 0.00\% \\
 & MATIC-A & Polygon & 194080 & 5193556 & 2676\% & 15000000 & 3.00\% & 175\% & 0.00\% \\
 & GNO-A & Gnosis & 3020632 & 5000000 & 166\% & 5000000 & 2.50\% & 350\% & 0.06\% \\ \midrule
& \multicolumn{2}{l}{\hspace{0.2cm}\textbf{Total (ERC-20 tokens):}} & \textbf{994893882} & \textbf{1407641598} & & & & & \textbf{20.00\%} \\ \midrule
\parbox[t]{2mm}{\multirow{5}{*}{\rotatebox[origin=c]{90}{LP Tokens}}} & UNIV2USDCETH-A & USDC/ETH& 2139902 & 0 & 0\% &  & 1.50\% & 120\% & 0.04\% \\
 & UNIV2DAIUSDC-A & Dai/USDC & 23983390 & 0 & 0\% &  & 0.02\% & 102\% & 0.48\% \\
 & GUNIV3DAIUSDC1-A & Dai/USDC & 59841202 & 0 & 0\% &  & 0.02\% & 102\% & 1.20\% \\
 & GUNIV3DAIUSDC2-A & Dai/USDC & 104141233 & 0 & 0\% &  & 0.06\% & 102\% & 2.09\% \\
 & CRVV1ETHSTETH-A & ETH/wETH & 56643346 & 66542992 & 117\% & 100000000 & 1.50\% & 155\% & 1.14\% \\ \midrule
& \multicolumn{2}{l}{\hspace{0.2cm}\textbf{Total (Liquidity Pool ERC-20 tokens):}} & \textbf{246749074} & \textbf{66542992} & & & & & \textbf{4.96\%} \\ \midrule
\parbox[t]{2mm}{\multirow{13}{*}{\rotatebox[origin=c]{90}{Real-World Assets}}} & RWA001-A & 6s Capital & 14348036 & 15000000 & 105\% &  & 3.00\% &  & 0.29\% \\
 & RWA002-A & Centrifuge & 6429638 & 20000000 & 311\% &  & 3.50\% &  & 0.13\% \\
 & RWA003-A & --- & 1824487 & 2000000 & 110\% &  & 6.00\% &  & 0.04\% \\
 & RWA004-A & --- & 1468187 & 7000000 & 477\% &  & 7.00\% &  & 0.03\% \\
 & RWA005-A & --- & 6076346 & 15000000 & 247\% &  & 4.50\% &  & 0.12\% \\
 & RWA006-A & AEA & 0 & 0 &  &  & 2.00\% &  & 0.00\% \\
 & RWA007-A & Monetalis & 499950318 & 1250000000 & 250\% &  & 0.00\% &  & 10.05\% \\
 & RWA008-A & SG Forge & 0 & 30000000 &  &  & 0.05\% &  & 0.00\% \\
 & RWA009-A & H.\ V.\ Bank & 100000000 & 100000000 & 100\% &  & 0.00\% &  & 2.01\% \\
 & RWA010-A & Centrifuge & 0 & 20000000 &  &  & 4.00\% &  & 0.00\% \\
 & RWA011-A & --- & 0 & 30000000 &  &  & 4.00\% &  & 0.00\% \\
 & RWA012-A & --- & 13619456 & 30000000 & 220\% &  & 4.00\% &  & 0.27\% \\
 & RWA013-A & --- & 69525268 & 70000000 & 101\% &  & 4.00\% &  & 1.40\% \\ \midrule
& \multicolumn{2}{l}{\hspace{0.2cm}\textbf{Total (Real-World Assets):}} & \textbf{713241736} & \textbf{1589000000} & & & & & \textbf{14.34\%} \\ \midrule
\parbox[t]{2mm}{\multirow{3}{*}{\rotatebox[origin=c]{90}{PSMs}}} & PSM-USDC-A & USD Coin & 2118503205 & 2608700392 & 123\% & 10000000000 &  &  & 42.58\% \\
 & PSM-USDP-A & Pax Dollar & 499997224 & 500000000 & 100\% & 500000000 &  &  & 10.05\% \\
 & PSM-GUSD-A & Gemini Dollar & 387538598 & 437538598 & 113\% & 500000000 &  &  & 7.79\% \\ \midrule
& \multicolumn{2}{l}{\hspace{0.2cm}\textbf{Total (Peg Stability Modules):}} & \textbf{3006039027} & \textbf{3546238990} & & & & & \textbf{60.42\%} \\ \midrule
    \end{tabular}
    }
\end{table*}

\section{Price Data Table}
\label{sec:price_data_table}
\Cref{tab:empirical_data_table} contains a specification of the time periods duirng which we collected price data for the empirical analysis of \Cref{sec:financial_analysis}. \Cref{fig:all_correlations} visualizes the correlations between the 100 tokens in our dataset.

\begin{table}[ht]
\centering
\caption{Data periods (times are in Greenwich Mean Time).}
\label{tab:empirical_data_table}
\begin{tabular}{c|ccc}
 \multicolumn{1}{c}{} & start & end & number of slots \\ \toprule
 Period 1 & 26-08-2022 02:16 & 30-08-2022 00:26 & 1131 \\
 Period 2 & 07-09-2022 03:15 & 15-09-2022 08:45 & 2371 \\
 Period 3 & 30-09-2022 07:37 & 02-10-2022 05:27 & 551\\
 Period 4 & 07-10-2022 04:31 & 15-10-2023 22:46 & 2010 \\
 Period 5 & 17-10-2022 02:51 & 29-10-2022 16:26 & 3620 \\
 Period 6 & 04-11-2022 03:50 & 23-11-2022 05:15 & 5490 \\ \bottomrule
\end{tabular}
\end{table}

\begin{figure*}[ht]
    \centering
	\includegraphics[width=\linewidth]{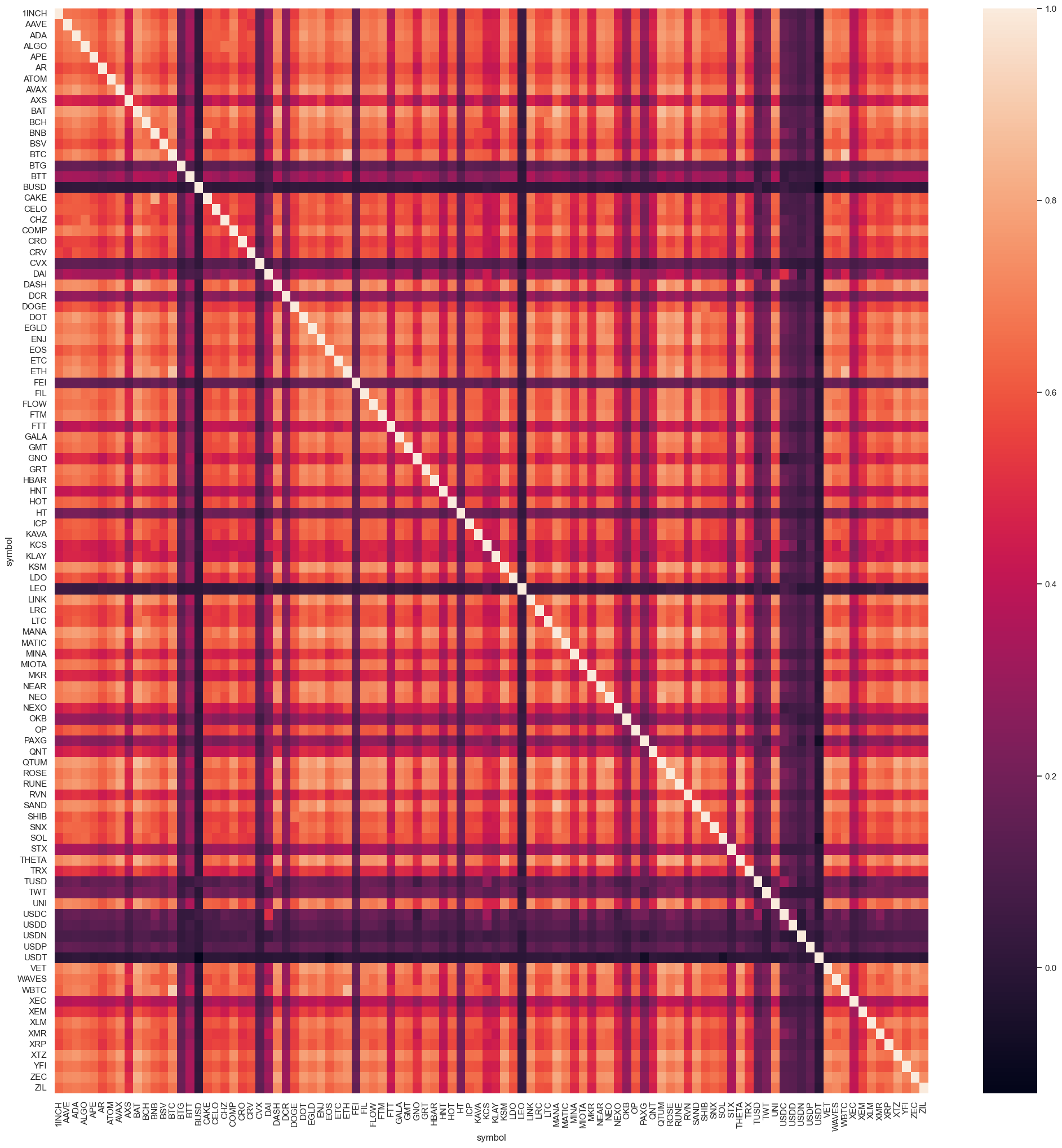}
	\caption{The correlations between all 100 tokens in our dataset.}
    \label{fig:all_correlations}
\end{figure*}


\section{Portfolio Compositions}
\label{sec:portfolio_compositions}
\Cref{tab:portfolio_compositions} displays the composition of the DSS-based simulation portfolios. \Cref{tab:optimal_portfolio_compositions} displays the composition of the optimized portfolios. \Cref{tab:portfolio_compositions_all} displays the optimal portfolio can be constructed using the 100 tokens in our dataset, if other stablecoins can be included.

\begin{table}[ht]
\caption{DSS-based portfolio compositions.}
\label{tab:portfolio_compositions}
\centering
\begin{tabular}{cc} 
portfolio name & token distribution \\
\toprule
\multirow{2}{*}{D-All} & 47\% USDC, 19\% ETH, 14\% PAXG \\ 
 & 10\% USDP, 8\% USDT, 2\% WBTC\\ \midrule
\multirow{2}{*}{D-ERC} & 90.2\% ETH, 9.4\% WBTC, 0.2\% GNO,  \\ 
 & 0.1\% LINK, 0.05\% MATIC, 0.05\% YFI\\ \midrule
\multirow{2}{*}{D-Mix1} & 50\% ETH, 30\% WBTC, 5\% GNO,  \\ 
 & 5\% LINK, 5\% MATIC, 5\% YFI\\ \midrule
\multirow{2}{*}{D-Mix2} & 20\% ETH, 20\% WBTC, 15\% GNO,  \\ 
 & 15\% LINK, 15\% MATIC, 15\% YFI\\ \bottomrule
\multirow{2}{*}{C-Mix1} & 50\% ETH, 30\% WBTC, 5\% ADA,  \\ 
 & 5\% DOT, 5\% TRX, 5\% AVAX \\ \bottomrule
\multirow{2}{*}{C-Mix2} & 20\% ETH, 20\% WBTC, 15\% ADA,  \\ 
 & 15\% DOT, 15\% TRX, 15\% AVAX \\ \bottomrule
\end{tabular}
\end{table}

\begin{table}[ht]
\caption{Optimal portfolio compositions.}
\label{tab:optimal_portfolio_compositions}
\centering
\begin{tabular}{cc} 
name & optimal token distribution \\
\toprule
\multirow{1}{*}{D-Opt} & WBTC 74.015\%, GNO 25.985\% \\ \midrule
\multirow{1}{*}{D-Opt} & WBTC 20\%, ETH 20\%  \\ 
($\lambda = 0.2$) & GNO 20\%, LINK 20\%, YFI 20\% \\ \midrule
\multirow{1}{*}{C-Opt} & WBTC 62.204\%, TRX 37.796\% \\ \bottomrule
\multirow{1}{*}{C-Opt} & WBTC 20\%, ETH 16.656\%, ADA 14.493\% \\
($\lambda = 0.2$) & DOT 20\%, TRX 20\%, EOS 8.851\% \\ \bottomrule
\multirow{4}{*}{A-Opt} & BSV 8.728\%, BTG 0.796\%, BTT 36.821\%, \\
& DCR 4.357\%, GNO 5.221\%, HT 0.0762\% \\
& KCS 2.726\%, KLAY 5.705\%, LEO 17.568\% \\
& OKB 3.223\%, TRX 13.017\%, TWT 0.483\%, XLM 0.593\% \\ \bottomrule
& BSV 10\%, BTG 1.354\%, BTT 10\%, \\
\multirow{2}{*}{A-Opt}& CAKE 0.644\%, DCR 7.278\%, GNO 10\%, \\
\multirow{2}{*}{($\lambda = 0.1$)}& HT 1.742\%, KCS 5.819\%, KLAY 9.434\%, \\
& LEO 10\%, NEXO 0.311\%, OKB 5.604\%, \\
& TRX 10\%, TWT 8.61\%, WBTC 10\%, XLM 6.951\% \\ \bottomrule
\end{tabular}
\end{table}

\begin{table}[ht]
\caption{Optimal simulation portfolio if other stablecoins are included.}
\label{tab:portfolio_compositions_all}
\centering
\begin{tabular}{ccc}
\toprule
BTT 0.34\% & BUSD 0.1504\% & PAXG 8.72\% \\ USDC 35.90\% & USDN 0.12\% & USDT 39.89\% \\
\bottomrule
\end{tabular}
\end{table}


\end{document}